\documentclass[10pt,journal]{IEEEtran}
\usepackage{amsmath,amssymb,amsfonts,amsthm,mathtools,bm}
\usepackage{float,enumitem}
\floatstyle{ruled}
\newfloat{algorithm}{tbp}{loa}
\floatname{algorithm}{Algorithm}
\usepackage{algpseudocode}
\usepackage{array,booktabs,multirow,makecell,tabularx}
\usepackage{cite}
\usepackage{url}
\usepackage[hidelinks]{hyperref}
\hypersetup{pdftitle={M-QAM MIMO Maximum-Likelihood Detection with QAOA},
  pdfauthor={Burhan G\"ulbahar}}
\newcommand{\aidoi}[1]{doi:\href{https://doi.org/#1}{#1}}
\newcommand{\aiarxiv}[1]{\href{https://arxiv.org/abs/#1}{arXiv:#1}}
\usepackage{graphicx}

\newtheorem{proposition}{Proposition}
\newtheorem{lemma}{Lemma}

\newtheorem{definition}{Definition}

\newcommand{\E}{\mathbb E}
\newcommand{\R}{\mathbb R}

\newcommand{\Acal}{\mathcal A}
\newcommand{\Pcal}{\mathcal P}
\newcommand{\Fcal}{\mathcal F}

\newcommand{\Cov}{\operatorname{Cov}}
\newcommand{\Var}{\operatorname{Var}}
\newcommand{\diag}{\operatorname{diag}}
\newcommand{\ii}{\mathrm i}

\newcommand{\Lcal}{\mathcal L}
\newcommand{\Ecal}{\mathcal E}
\newcommand{\Rea}{\widehat{O}_{E,a}}
\newcommand{\Reb}{\widehat{O}_{E,b}}
\newcommand{\EbBmid}{3.2}
\newcommand{\EbBend}{2.1}
\newcommand{\rpl}{\hat r_{\rm pl}}
\newcommand{\rml}{\hat r_{\rm ML}}
\newcommand{\rs}{\hat r_{s}}
\newcommand{\Bsf}{\mathsf B}

\begin{document}

\title{M-QAM MIMO Maximum-Likelihood Detection with QAOA:
ML-Rate Offline Angle Design and Correlated Infinite-Size
Spin-Glass Models}

\author{Burhan G\"ulbahar
\thanks{B. G\"ulbahar is with the Department of Electrical and Electronics Engineering, Ya\c{s}ar University, 35100 Izmir, Turkey (e-mail: burhan.gulbahar@yasar.edu.tr).}
}

\maketitle

\begin{abstract}
The quantum approximate optimization algorithm (QAOA) targets NP-hard
maximum-likelihood (ML) detection in multiple-input multiple-output (MIMO)
systems. Existing $M$-ary quadrature amplitude modulation (M-QAM) detectors
design angles by expected Ising energy: online per instance, warm-started, or
ramped, while train-once designs remain B/QPSK-only or block-local, leaving
M-QAM without a size-scalable benchmark. Their infinite-size spin-glass
theory assumes independent disorder, matching the retained
covariances at B/QPSK but not M-QAM's
correlated couplings and fields. We develop a correlated infinite-size
multi-species spin-glass framework whose covariance-matched evaluators make
that energy an offline objective with a size-scalable benchmark. In addition,
the ML rate, the exponential rate of sampling the ML string, is for the
first time exploited for QAOA angle design in MIMO detection. The energy
evaluator is \(q\)-free at \(O(p \,4^p)\) cost while the ML rate transfers
angles from a fixed \(q_{\rm ref}\)-qubit reference. Tests reach 4096-QAM,
128 antennas, \(p=30\) and per-symbol SNR 0--45\,dB. In simulations, ML rates fall as a power law \(r_0\,p^{-\alpha}\), with larger
exponents for the sampling design, which tracks exact ML at \(5\times5\)
16-QAM (0--20\,dB) and \(3\times3\) 64-QAM (8--28\,dB) while its bit-error
rate (BER) advantage widens with SNR to two orders of magnitude. The approach points toward near-optimum decoding on deeper noiseless
fault-tolerant quantum (FTQ) circuits.
\end{abstract}


\section{Introduction}

Maximum-likelihood (ML) detection of \(M\)-ary quadrature amplitude
modulation (\(M\)-QAM) in large multiple-input multiple-output (MIMO)
systems is NP-hard in the worst case, while tractable linear receivers
remain suboptimal~\cite{verdu1989,singh2022ising,yang2015fifty}. The quantum
approximate optimization algorithm (QAOA)~\cite{farhi2014quantum} instead
samples a depth-\(p\) Ising state, promising near-optimum MIMO
detection performance~\cite{gulbahar2024mimo}.
Available
detectors, however, optimize angles per channel
online~\cite{cui2022qaoa,cui2022general}, add per-channel warm starts to a
linear-ramp family~\cite{paul2026warmstart}, or reuse an offline bank across
QR-domain blocks regularized by minimum mean-square error
(MMSE)~\cite{zeng2026block}. Each tunes angles to the expected Ising energy rather than to the rate at
which the ML string is sampled, with results reported for specific
modulations and size classes. A train-once cross-size \(M\)-QAM
prescription, an angle objective built on that ML rate, and a
prescribed-depth, finite-\(q\)-independent performance benchmark (\(q\) is the qubit
count) remain unavailable.
Ising-machine detectors already reach near-ML accuracy at high-order
QAM through dedicated regularized and perturbation-based
formulations~\cite{singh2022ising} while  we target gate-model circuits with angles fixed offline. The fault-tolerant quantum (FTQ) era, rather than near-term noisy hardware,
is targeted here, supporting the deep, noiseless QAOA circuits that ML-level
sampling needs; IBM projects Quantum Starling by 2029, running \(10^8\) gates
on \(200\) logical qubits~\cite{ibm2025ftqc}. Circuits are accordingly
all-to-all, leaving logical synthesis of the continuous angle rotations
\((\gamma,\beta)\) and its distillation overhead outside the present scope.
This article aims at that regime: angles designed
on the ML rate instead of the Ising energy, that rate decaying as a power law
$p^{-\alpha}$ in depth, and frozen angle packs that transfer across system sizes. Simulations support near-optimal \(M\)-QAM decoding  while a
matching large-size theory remains open. 
Both routes are size-scalable in design cost: the correlated evaluator
\(V_p^{\rm corr}\) spends \(4^p\) paths at \(O(p \,4^p)\) cost whatever the
antenna count is, so one evaluation serves every array size, while the ML
rate is paid once per pack, on a frozen reference at \(q_{\rm ref}\) under
classical evaluation and transferred, or at the deployed size directly when
shots are available. Only the energetic route also provides a size-independent
evaluator of its objective, making a prescribed-depth benchmark practical; the
corresponding benchmark for the sampling  remains open.

The earlier BPSK/QPSK work transfers infinite-size independent
Sherrington--Kirkpatrick (SK) angles~\cite{gulbahar2024mimo,gulbahar2025mvsic},
but omits M-QAM's weighted correlations for $n_t \times n_r$ MIMO system: same-axis intra-symbol couplings have
weight-dependent \(O(n_r)\) means, R/I counterparts vanish, and fields
correlate with incident couplings. Independent SK and
Basso--Gamarnik--Mei--Zhou (BGMZ) recurrences lack fixed-level \(V_p\) theory
for this axis/bit/weight
structure~\cite{farhi2022sk,bgmz2022constant,basso2022highdepth}; the spin--boson route~\cite{boulebnane2025spinboson} also excludes the block-Wishart ensemble, i.e., couplings read off the
Gram matrix \(\widetilde H^{\!\top}\widetilde H\) of the realified Gaussian
channel \(\widetilde H\) (a Wishart matrix) weighted and grouped into
per-symbol blocks by the bit-to-symbol map.  We derive this geometry and test covariance-matched
surrogates.

\subsection{Contributions}\label{sec:contributions}

\emph{Offline} means ensemble/reference-trained without per-instance
reoptimization; transfer keeps modulation, signal-to-noise ratio (SNR)
convention, and aspect ratio
fixed~\cite{brandao2018fixed,galda2021transferability,shaydulin2022parameter}.
\begin{itemize}[leftmargin=*,itemsep=1pt,topsep=2pt,parsep=0pt]
\item \emph{Correlation-aware fixed-depth framework.} Fixed-depth
infinite-size QAOA theory covers ensembles with i.i.d.\ couplings, i.e., independent
SK, pure and mixed
$k$-spin~\cite{farhi2022sk,bgmz2022constant,basso2022highdepth}, whose $V_p$
recurrences also assume no field--coupling correlation; $M$-QAM satisfies
neither condition. We supply the missing evaluator as a covariance-matched
multi-species action: the scalar form uses $4^p$ paths, independent of the
transmit-antenna count $n_t$ and of $q$, with exact $O(p\,4^p)$ fast
Walsh--Hadamard transform (FWHT) readout; a bit-block refinement ($4^{p\,b}$
paths for $b$ bits per symbol) is left for future work. It is a structural surrogate, not
the proved correlated-Wishart $V_p$ limit, which remains open. However, no earlier
fixed-depth evaluator carries these projected correlations.
\item \emph{Exact coefficient geometry.} BPSK/QPSK
transfers~\cite{gulbahar2024mimo,gulbahar2025mvsic} reuse independent SK
angles, whose retained covariances are those of independent disorder, a
match that holds for the projected moment family at BPSK/QPSK and
fails at $M$-QAM.
Proposition~\ref{prop:moments} gives: deterministic
intra-symbol blocks, zero R/I partners, centered inter-symbol Wishart terms,
weighted fields, and incident $J$--$h$ correlation.
\item \emph{Energy versus sampling.} Two offline objectives are compared on
identical instances: minimizing the correlated energy surrogate
$V_p^{\rm corr}$ of $V_p$, and minimizing the ML rate (\(-q^{-1}\) times the
mean log ML-word mass), which has not been exploited for MIMO angle design. Both
scale, but differently: $V_p^{\rm corr}$ is $q$-free by construction, whereas
the ML rate is paid once at $q_{\rm ref}$ and transferred. Their annealed rates $r_0 \, p^{-\alpha}$ nearly coincide at $p=1$ but separate into $\alpha=0.162$ and $0.451$ for 16-QAM at $q=20$, 20\,dB, and sampling
reaches exact-ML floors at depths the energy route does not: resolvable
bit-error-rate (BER) gaps of $142\times$ at $p=4$ here and $36\times$ at
$p=6$ for 64-QAM at $q=18$, 28\,dB. With $N=1000$ instances, sampling becomes indistinguishable from exact ML
at the $1/(Nq)$ resolution from $p=5$ (16-QAM) and $p=10$ (64-QAM).
\end{itemize}
Table~\ref{tab:symbols} lists recurring symbols; others are defined
locally. In Table~\ref{tab:related}, objective E/R means expected energy/ML-rate, and $\dagger/\ddagger$
denote $V_p/V_p^{\rm corr}$; ``$M$''  and ``Max. $n_t{\times}n_r$'' denote tested detector-performance
constellations and maximum system sizes, respectively.  ``Bench.''
marks a size-independent performance evaluator of the objective,
which the ML-rate route still lacks.
Section~\ref{sec:related} places the work among quantum and physics-inspired
detectors and Section~\ref{sec:qaoa_sk} fixes the QAOA setting;
Sections~\ref{sec:mqam_model}--\ref{sec:infsize} derive the M-QAM coefficient
geometry, the comparison metric, and the two offline angle objectives.
Section~\ref{sec:numerics} reports simulations; Section~\ref{conclusion}
concludes.

\begin{table}[t]
	\centering
	\caption{Quantum and physics-inspired M-QAM MIMO ML detectors}
	\label{tab:related}
	\begingroup
	\footnotesize
	\setlength{\tabcolsep}{0.8pt}
	\setlength{\arrayrulewidth}{0.3pt}
	\renewcommand{\arraystretch}{1.0}
	\begin{tabular}{|
			>{\raggedright\arraybackslash}p{0.136\columnwidth}|
			>{\raggedright\arraybackslash}p{0.126\columnwidth}|
			>{\centering\arraybackslash}p{0.194\columnwidth}|
			>{\centering\arraybackslash}p{0.175\columnwidth}|
			>{\centering\arraybackslash}p{0.16\columnwidth}|
			>{\centering\arraybackslash}p{0.058\columnwidth}|
			>{\centering\arraybackslash}p{0.091\columnwidth}|}
		\hline
		Ref. & Method & $M$ & \makecell[c]{Max. $n_t{\times}n_r$} &
		Ang. & Obj. & \makecell[c]{Bench.} \\
		\hline
		\multicolumn{7}{|l|}{\emph{Search, annealing, and Ising machines}}\\
		\hline
		\cite{botsinis2013,kim2019quamax,kim2025xresq,kim2021paramax}
		& \makecell[l]{QWSA,\\QA/RA,\\PT} & $\{2,4,16\}$ &
		$1024\times1024$ & --- & --- & --- \\
		\hline
		\cite{singh2022ising,jamieson2025mmgap,zhu2026pimi}
		& \makecell[l]{CIM/PIC,\\PIM} & $\{2,4,16,64\}$ &
		$64\times256$ & --- & --- & --- \\
		\hline
		\multicolumn{7}{|l|}{\emph{Gate-model QAOA}}\\
		\hline
		\cite{cui2022qaoa} & --- & $2$ & $3\times3$ & online & E & --- \\
		\hline
		\cite{cui2022general} & gen. H & $4$ & $1\times1$ & online & E & --- \\
		\hline
		\cite{gulbahar2024mimo} & SK & $\{2,4\}$ &
		$128\times128$ & offline & E$^{\dag}$ & \checkmark \\
		\hline
		\cite{gulbahar2025mvsic} & RQAOA & $\{2,4\}$ &
		$64\times64$ & offline & E$^{\dag}$ & \checkmark \\
		\hline
		\cite{gulbahar2025rrtq} & RRT-Q & $\{2,4\}$ &
		$256\times256$ & offline & E$^{\dag}$ & \checkmark \\
		\hline
		\cite{paul2026warmstart} & WS & $16$ & $4\times4$ &
		\makecell[c]{online ramp} & E & --- \\
		\hline
		\cite{zeng2026block} & block & $16$ & $16\times16$ &
		offl. bank & E & --- \\
		\hline
		\multirow{2}{\linewidth}{\textbf{This work}} &
		\multirow{2}{\linewidth}{\textbf{corr. MS}} &
		$\mathbf{\{16,64\}}$ &
		$\mathbf{128\times128}$ & \textbf{offline} &
		\textbf{E}$^{\ddag}$ & \checkmark \\
		\cline{3-7}
		& & \makecell[c]{$\mathbf{\{16,64,128,}$\\$\mathbf{256,1024,}$\\$\mathbf{4096\}}$}
		& \makecell[c]{$\mathbf{6\times6}$,\\$\mathbf{2\times48}$} & \textbf{offline} &
		\textbf{R} & --- \\
		\hline
	\end{tabular}
	\endgroup
\end{table}

\section{Related Works}\label{sec:related}
Among the cited gate-model MIMO studies in Table~\ref{tab:related}, angle
selection uses expected energy (online, warm-started, ramp-selected, banked,
or through $V_p$) rather than the ML sampling rate.
A quantum weighted sum algorithm (QWSA) is tested on QPSK and
8-PSK~\cite{botsinis2013}. Quantum annealing
(QA)~\cite{kim2019quamax}, reverse annealing
(RA)~\cite{kim2025xresq} and parallel tempering (PT)~\cite{kim2021paramax}
reach the largest arrays in the table, and coherent Ising machines
(CIM)~\cite{singh2022ising}, physics-inspired computation
(PIC)~\cite{jamieson2025mmgap} and a probabilistic Ising machine with inertia
(PIM)~\cite{zhu2026pimi} extend them to higher orders. All provide
finite-size evidence rather than ensemble M-QAM limits, and their largest
entries are classical or sampling-capacity studies rather than
detector-performance runs.
Gate-model studies use plain online BPSK/QPSK angles \cite{cui2022qaoa}, a general Hamiltonian (gen.\ H)
construction \cite{cui2022general}, per-channel M-QAM warm starts (WS)
selected over a fixed linear-ramp family~\cite{paul2026warmstart}, or
SNR-indexed block banks with parameter transfer~\cite{zeng2026block}. The
earlier B/QPSK detectors follow the energy route with offline angles, through
the infinite-size SK
model~\cite{gulbahar2024mimo}, majority voting with recursive QAOA
(RQAOA)~\cite{gulbahar2025mvsic}, and recursive-repetitive truncated QAOA
(RRT-Q)~\cite{gulbahar2025rrtq}; M-QAM is not modeled there, as the required
correlated infinite-size spin-glass model was unavailable. This work designs
against the correlated multi-species (corr.\ MS) ensemble of
Section~\ref{sec:mqam_model}.
To our knowledge, none of the cited works gives a full-system,
correlation-aware, size-independent M-QAM benchmark.

\section{QAOA ML Detection}\label{sec:qaoa_sk}
For sign vector $z=(z_1,\ldots,z_q)\in\{\pm1\}^q$, the Ising cost is
\begin{equation}
 C(z)=\sum_{i<j}J_{ij}z_i z_j+\sum_i h_i z_i + C_{const}
 \label{eq:Ising_cost}
\end{equation}
where $J_{ij}$ and $h_i$, with bit indices $i,j\in\{1,\dots,q\}$ and
couplings only for $i\ne j$, collect into $J=(J_{ij})\in\R^{q\times q}$
(symmetric, zero diagonal) and $h=(h_i)\in\R^{q}$, while
$C_{const}$ shifts the spectrum uniformly. Therefore, the Hamiltonian becomes
$H_C=\sum_{i<j}J_{ij}Z_i Z_j+\sum_i h_i Z_i+C_{const}$ with Pauli-Z operators.
From uniform input $|+\rangle^{\otimes q}$,
QAOA uses operators
$U_C(\gamma)=e^{-\ii\gamma H_C}$ and
$U_B(\beta)=e^{-\ii\beta B}$, \(B=\sum_iX_i\) with Pauli-X operators, and prepares
\(\lvert\gamma,\beta\rangle\triangleq 
U_B(\beta_p)U_C(\gamma_p)\cdots U_B(\beta_1)U_C(\gamma_1)
\lvert+\rangle^{\otimes q}\)  \cite{farhi2014quantum}. Its
\(2p\) angles minimize \(\langle H_C\rangle\) or  \(\langle C \rangle\) where  
\(\langle O\rangle=\langle\gamma,\beta|O|\gamma,\beta\rangle\)  \cite{gulbahar2024mimo,gulbahar2025mvsic}.

\section{M-QAM: Correlated  Ising Model}\label{sec:mqam_model}
Representing complex channel matrix as  $H=H_R+\ii H_I$  gives the following
\cite{gulbahar2024mimo} where $y_{\rm c}=Hs_{\rm c}+n_{\rm c}$ and $\rho\triangleq n_r/n_t$ for $n_t$ transmit and $n_r$ receive
antennas:
\begin{equation}
 \widetilde H=\begin{bmatrix} H_R&-H_I\\ H_I&H_R\end{bmatrix};\,
  \widetilde y=\begin{bmatrix} y_R\\ y_I\end{bmatrix}; \,
  \widetilde z=\begin{bmatrix} z_R\\ z_I\end{bmatrix}=T \,z
 \label{eq:realification}
\end{equation}
where M-QAM has $L_P=2^{b_P}$ levels for $P\in\{R,I\}$, $M=L_RL_I$, and
$q=n_t(b_R+b_I)$ with $\widetilde H\in\R^{2n_r\times2n_t}$,
$T\in\R^{2n_t\times q}$:
\begin{equation}
 T=\begin{bmatrix}
 2^{b_R-1}I_{n_t}&2^{b_R-2}I_{n_t}&\cdots&I_{n_t}&0\\
 0&2^{b_I-1}I_{n_t}&2^{b_I-2}I_{n_t}&\cdots&I_{n_t}
 \end{bmatrix}
 \label{eq:T_matrix}
\end{equation}
For transmitted
$s\in\{\pm1\}^q$ and $\widetilde n\triangleq[n_R;n_I]$,
$\widetilde y=\widetilde HTs+\widetilde n$,  ML solution is obtained as follows:
\begin{equation}
 \min_{z\in\{\pm1\}^q}\|\widetilde y-\widetilde H Tz\|^2 
 \label{eq:M_QAM_binary_ML}
\end{equation}
Expansion gives
$C(z)=z^TGz-2z^TT^T\widetilde H^T\widetilde y+C_0$  where $C_0=\lVert\widetilde y\rVert^2$,  and for $i<j$, the following is obtained:
\begin{align}
 J_{ij}=2 \,G_{ij}; \, h_i=-2(T^T\widetilde H^T\widetilde y)_i;\, G=T^T\widetilde H^T\widetilde H T
 \label{eq:M_QAM_Jh}
\end{align}
Here $G$ is the $T$-weighted Wishart Gram for Gaussian $\widetilde H$.
\subsection{Correlated Moment Structure and
\texorpdfstring{$A,\Bsf,D$}{A, B, D} Parameters}\label{sec:moments}

Index a bit by $i=(u_i,P_i,r_i)$, $i\in\{1,\dots,q\}$: symbol
$u_i\in\{1,\dots,n_t\}$, component $P_i\in\{R,I\}$, bit position
$r_i\in\{1,\dots,b_{P_i}\}$, and weight
$w_i=2^{b_{P_i}-r_i}$; couplings $J_{ij}$ always carry $i\ne j$. Let
$\Ecal_P\triangleq\sum_rw_{P,r}^2=(L_P^2-1)/3$, $E_s=\Ecal_R+\Ecal_I$ (mean symbol energy),
$b=b_R+b_I$, and $q=n_tb$.
The structured realification \eqref{eq:realification} has per-entry variance
$\sigma_H^2/2$ (complex variance $\sigma_H^2$), and $\widetilde n$ has i.i.d.\
variance $\sigma_n^2/2$. We distinguish
$\mathrm{SNR}_{\rm ps}=E_s/\sigma_n^2$ and
$\mathrm{SNR}_{\rm rx}=n_t\sigma_H^2E_s/\sigma_n^2$, offset by
$10\log_{10}(n_t\sigma_H^2)$ ($10\log_{10}n_t$ here since
$\sigma_H^2{=}1$); quoted SNRs are per-symbol unless stated
otherwise.
Unless subscripted, \(\E,\Var,\Cov\) in this
section are over the joint law of \(H,\widetilde n,s\); conditioning on
\(s\) fixes the word.
Entries of $J$ and $h$ are $2n_r$-term sums of $O(1)$ products
\eqref{eq:M_QAM_Jh}, so they scale as $O(n_r)$; at fixed
$n_r/n_t$ and fixed $b$, $n_r\propto n_t=q/b$, hence $O(n_r)=O(q)$.
Each normalized $\bar J_{ij},\bar h_i$ is $O_p(1)$ after
division by $q$ (the coefficient scale, not the total eigenphase, is
size-stable) so
one angle pack can be applied at different sizes:
\begin{equation}
 \bar C_q(z)=\sum_{i<j}\bar J_{ij}z_i z_j+\sum_i\bar h_i z_i;\,
 \bar J_{ij}=J_{ij}/q;\,\ \bar h_i=h_i/q
 \label{eq:normalized_coeffs}
\end{equation}
This only rescales QAOA phases,
$e^{-\ii\gamma\bar C_q}=e^{-\ii(\gamma/q)C}$ up to a global phase, and leaves
the minimizer unchanged.
Proposition~\ref{prop:moments} summarizes the structure. For realified channel column \(v_{u,P}\) with
$v_{u,R}=[H_R(:,u);H_I(:,u)]$ and $v_{u,I}=[-H_I(:,u);H_R(:,u)]$, define
\(g_{u,P}\triangleq v_{u,P}^{T}\widetilde y\), an entry of the matched filter
\(\widetilde H^{T}\widetilde y\) of \eqref{eq:M_QAM_Jh} that correlates the
received vector against the channel columns. Its distribution is
\(u\)-independent but differs between the \(R\) and \(I\) axes. Averaging
over the transmitted word cancels the signal part, so \(g_{u,P}\) is centered
and $\Lcal_P\triangleq \Var(g_{u,P})=\E[g_{u,P}^{2}]$ is its per-axis
variance.
All bits on the same \((u,P)\) share it, and
\(h_i=-2w_i g_{u_i,P_i}\); thus \(\Lcal_P\) is the common factor in
\(\Var(h_i)=4w_i^2\Lcal_{P_i}\) and the corresponding same-axis field
covariances.
For \(i,j,k\in\{1,\ldots,q\}\), define the uniform index means
\(\E_i f_i\triangleq q^{-1}\sum_i f_i\);
\(\E_{i<j}f_{ij}\triangleq \binom q2^{-1}\sum_{i<j}f_{ij}\);
\(\E_{i\ne j}f_{ij}\triangleq [q(q-1)]^{-1}\sum_{i\ne j}f_{ij}\); and
\(\E_{(i,j,k)}f_{ijk}\triangleq [q(q-1)(q-2)]^{-1}
\sum_{|\{i,j,k\}|=3}f_{ijk}\), the final sum being over ordered triples.

\begin{proposition}[Correlated moments of M-QAM]\label{prop:moments}
Let the physical channel $H$ have i.i.d.\ proper-complex Gaussian entries of
variance $\sigma_H^2$,  noise $n$ be proper-complex Gaussian and
independent of $H$, and  transmitted signs $s$ be uniform and
independent of both; $\widetilde H,\widetilde n$ are real parts
\eqref{eq:realification}, whose $R$/$I$ partner columns are therefore
orthogonally paired rather than independent. Then, the couplings have the
following mean:
\begin{equation}
\E[J_{ij}]=2n_r\sigma_H^2w_iw_j\ \ \text{for }u_i{=}u_j,\,P_i{=}P_j,\,r_i{\ne}r_j
\label{eq:J_mean_MQAM}
\end{equation}
(zero otherwise) and variance as follows:
\begin{equation}
\Var(J_{ij})=
\begin{cases}
4n_r\sigma_H^4 w_i^2w_j^2, & u_i{=}u_j,\;P_i{=}P_j,\;r_i{\ne}r_j,\\
0\ \text{(identically)}, & u_i{=}u_j,\;P_i{\ne}P_j,\\
2n_r\sigma_H^4 w_i^2w_j^2, & u_i{\ne}u_j
\end{cases}
\label{eq:J_var_MQAM}
\end{equation}
The fields have $\Var(h_i)=4w_i^2\Lcal_{P_i}$ for every $i$ and
$\Cov(h_i,h_j)=4w_iw_j\Lcal_{P_i}$ on the class of
\eqref{eq:J_mean_MQAM}, zero otherwise.
Conditioned on the transmitted word, we have the following:
\begin{equation}
 \E[h_i\mid s]=-2n_r\sigma_H^2w_i\,x_{u_i,P_i}; \, \,
  x_{u,P}=\sum_r w_{P,r}s_{u,P,r} 
 \label{eq:h_cond_mean}
\end{equation}
and, for inter-symbol edge $u_i\ne u_j$, the following is obtained:
\begin{equation}
\begin{aligned}
\Cov(J_{ij},h_i\mid s)&=-2n_r\sigma_H^4\,w_i^2w_j\,x_{u_j,P_j}
\end{aligned}
\label{eq:D_cond_MQAM}
\end{equation}
With \eqref{eq:normalized_coeffs} and the centered coefficients
$\bar J^c_{ij}=\bar J_{ij}-\E[\bar J_{ij}]$,
$\bar h_i^c=\bar h_i-\E[\bar h_i]$, the scalar exchangeable projection is as follows:
\begin{eqnarray}
4\sigma_J^2 \triangleq q\,\E_{i<j}\big[\E[(\bar J^c_{ij})^2]\big]
=\tfrac{1}{q}\,\E_{i<j}[\Var(J_{ij})] && \label{eq:scalar_variances}\\
\sigma_h^2 \triangleq \E_i\big[\E[(\bar h_i^c)^2]\big]
=\tfrac{1}{q^2}\,\E_i[\Var(h_i)] && \label{eq:scalar_sigmah}\\
A \triangleq q^2\,\E_{(i,j,k)}[\Cov(\bar J^c_{ij},\bar J^c_{ik})]
=\E_{(i,j,k)}[\Cov(J_{ij},J_{ik})] && \label{eq:scalar_A}\\
\Bsf \triangleq q\,\E_{i\ne j}[\Cov(\bar h_i^c,\bar h_j^c)]
=\tfrac{1}{q}\,\E_{i\ne j}[\Cov(h_i,h_j)] && \label{eq:scalar_B}\\
D \triangleq q\,\E_{i<j}[\Cov(\bar J^c_{ij},\bar h_i^c)]
=\tfrac{1}{q}\,\E_{i<j}[\Cov(J_{ij},h_i)]=0 && \label{eq:scalar_D}
\end{eqnarray}
The raw forms are closed in $n_r,n_t,\sigma_H^2,\sigma_n^2$ and QAM weights.
The symmetric sign average gives $\E[h_i]=\Cov(J_{ij},h_i)=0$ resulting in $D = 0$; conditioning on the word retains
\eqref{eq:D_cond_MQAM}.
\end{proposition}
\begin{proof}[Proof sketch]
With realified columns $a_i=w_iv_{u_i,P_i}$, every entry follows from
sign-averaged Gaussian moment factorization over the i.i.d.\ columns and
shared-index counting:
$\E\|v\|^2=n_r\sigma_H^2$,
$\E\|v\|^4=n_r(n_r{+}1)\sigma_H^4$,
$\E(v^Tv')^2=n_r\sigma_H^4/2$, and
$\E(v^T\widetilde n)^2=n_r\sigma_H^2\sigma_n^2/2$ give the variances and the
closed $\Lcal_P=n_r\sigma_H^4[(n_r{+}1)\Ecal_P+(n_t{-}1)E_s/2]
+n_r\sigma_H^2\sigma_n^2/2$; only same-symbol/same-axis classes carry means;
pairs sharing exactly one index produce the $A,\Bsf,D$ classes at the printed
$q$-scalings; and centering commutes with the $q^{-1}$ limit.  Write $S_{m,P}\triangleq\sum_{r=1}^{b_P}w_{P,r}^m$,
$S_{2,P}=\Ecal_P$, $E_s=S_{2,R}+S_{2,I}$, and
$\Delta_P\triangleq S_{1,P}^2-S_{2,P}$. With $n_r/n_t\to\rho$,
$\sigma_n^2/n_t\to\nu_{\rm n}$, and
$\ell_P\triangleq\lim\Lcal_P/n_t^2
=\rho\sigma_H^4(\rho S_{2,P}{+}E_s/2)
+\rho\sigma_H^2\nu_{\rm n}/2$. Summing over $P\in\{R,I\}$ in
\eqref{eq:scalar_variances}--\eqref{eq:scalar_D} gives
\begin{equation}
\begin{alignedat}{2}
\sigma_J^2&=\frac{\rho \, \sigma_H^4E_s^2}{2b^3};\,&
\sigma_h^2&=\frac4{b^3}\sum_P S_{2,P}\ell_P;\\[-2pt]
A&=\frac{2 \, \rho \, \sigma_H^4E_s}{b^3}\sum_P\Delta_P;\,&
\Bsf&=\frac4{b^3}\sum_P\Delta_P\ell_P;\,D=0.
\end{alignedat}
\label{eq:scalar_moment_limits}
\end{equation}
Here $\Delta_P=S_{1,P}^2-S_{2,P}=2\sum_{r<t}w_{P,r}w_{P,t}\ge0$, with
equality iff $b_P\le1$; hence $A,\Bsf\ge0$, and $A=\Bsf=0$
exactly at BPSK/QPSK: the independent-disorder theory
\cite{gulbahar2024mimo} is the degenerate corner of the present one within
this retained family; identities such as
$J_{(u,R,r),(v,R,t)}{=}J_{(u,I,r),(v,I,t)}$ tie disjoint pairs at every
constellation and lie outside, one instance of what the exchangeable
projection of Sec.~\ref{sec:infsize} drops.
\end{proof}

\section{Normalization and the Comparison Metric}
\label{sec:normalization}
 
The diagonal and \(C_0\) give only
\(C_{const}=C_0+\operatorname{Tr}G\) since \(z_i^2=1\), so
\eqref{eq:normalized_coeffs} has raw-scale form
$q \, \bar C_q(z)=\sum_{i\ne j}G_{ij}z_iz_j+h^\top z$.
Proportional intra-symbol columns make the raw cost \(O(q^2)\), so
\(\bar C_q\) grows linearly in \(q\) and \(\langle\bar C_q\rangle/q\) is the
size-stable density (Fig.~\ref{fig:largen}).  $U_C$ denotes
the phase separator built from \(\bar C_q\), the QAOA unitary that applies
the cost as a phase:
\begin{equation}
	U_C(\gamma)=\exp\!\big(-\ii\,\gamma\,\bar C_q\big),\qquad \gamma=O(1)
	\label{eq:phase_q}
\end{equation}

\subsection{An offset-free figure of merit}

Denote by \(\E_z\triangleq \E_{z\sim{\rm Unif}(\{\pm1\}^q)}\) the uniform
average over words and by $C_{\rm ML}\triangleq\bar C_q(\hat z)$ the value at
the ML word $\hat z$ minimizing \eqref{eq:M_QAM_binary_ML}. The uniform reference then vanishes:
\begin{align}
 C_{\rm rand}&=\E_z[\bar C_q(z)]\nonumber\\
 &=q^{-1}\textstyle\sum_{i\ne j}G_{ij}\,
 \E_z[z_iz_j]+q^{-1}h^\top\E_z[z]=0
 \label{eq:crand_zero}
\end{align}
\begin{equation}
 R_p=\frac{\langle\bar C_q\rangle-C_{\rm ML}}{C_{\rm rand}-C_{\rm ML}}
 =1-\frac{\langle\bar C_q\rangle}{C_{\rm ML}}
 \label{eq:Rp}
\end{equation}
Therefore, $R_p=0$ at ML and $R_p=1$ at random, the second form needing
$C_{\rm ML}\neq0$, which holds since $C_{\rm ML}\le C_{\rm rand}=0$ with
equality only for constant cost; only states no worse than random
lie in $[0,1]$, and averages use per-instance ratios. The transmitted-word $s$
value $C_{\rm tx}\triangleq \bar C_q(s)$ serves as the second energy
reference reported in Figs.~\ref{fig:objective} and~\ref{fig:largen}.

\section{Correlation-Aware Energy Surrogate and Sampling Law} \label{sec:infsize}
This section defines the four design objectives, constructs analytic
energy surrogate \(O_E\) and the general-depth sampling law for
\(O_s\), including the fixed-\(p\) determinant model and efficient \(p=1\)
form, and states the offline search protocol and the shot budget relating sampled BER to ML mass.
\subsection{Objectives and Offline Angle Design}\label{sec:design_objectives}
The energetic route maps the
physical moments  to $O_E$ and its finite
approximations $\Rea,\Reb$ while the sampling one combines the exact
fixed-$p$ determinant law, its efficient $p=1$ specialization, and the
general-depth objective $O_s$. Both use the continuation protocol of Section~\ref{sec:search_protocol}, with
exact ensemble evaluation practical only at $p=1$. For $\theta=(\gamma,\beta)\in\mathbb R^{2p}$, write
\(\widehat\E_{\mathcal I}f\triangleq|\mathcal I|^{-1}\sum_{I\in\mathcal I}f(I)\)
for the empirical average over a frozen $N_{\rm tr}$-instance reference
ensemble $\mathcal I$ with $q_{\rm ref}=n_t^{\rm ref}(b_R+b_I)$, and let
\(P^{\rm ML}_{q_{\rm ref},p}(I;\theta)\triangleq|\langle\hat z_I|\gamma,\beta\rangle|^2\)
be the ML mass of instance \(I\)'s depth-\(p\) QAOA state. The design rate is the per-qubit exponent of this mass: under the
large-$q$ scaling assumption $P^{\rm ML}=e^{-q\varrho(\theta)+o(q)}$,
the logarithm makes the objective
intensive and   equal to the per-qubit log shot count of $O(1/P^{\rm ML})$ sampling
\cite{boulebnane2024sat}:
\begin{equation}
 \rs(\theta;\mathcal I)=-\frac{1}{q_{\rm ref}|\mathcal I|}
 \sum_{I\in\mathcal I}\log P^{\rm ML}_{q_{\rm ref},p}(I;\theta) 
 \label{eq:rs_def}
\end{equation}
written \(\rs(\theta)\) when \(\mathcal I\) is fixed. The design objective
minimized by route $\mathrm O$ is then as follows:
\begin{equation}
 \mathcal J_{\mathrm O}(\theta)=
 \begin{cases}
 \Re\,V_p^{\rm corr}(\theta) & O_E,\\[2pt]
 \widehat{\E}_{\mathcal I}[\langle\bar C_{q_{\rm ref}}\rangle_\theta]
   & \Rea,\Reb,\\[2pt]
 \rs(\theta;\mathcal I) & O_s
 \end{cases}
 \label{eq:objective_definitions}
\end{equation}
Only the last two rows use $\mathcal I$: $V_p^{\rm corr}$ is the
infinite-size evaluator and needs no instances.
The real part $\Re$ in \eqref{eq:objective_definitions} is definitional:
reality of $V_p^{\rm corr}$ on the selected branch is not established,
the underlying stationary point is complex, and branch uniqueness is
not established; the continuation fixes the branch, and $\Re$ makes
the selected value real.
\subsection{Evaluators for \(\Rea\) and \(\Reb\), and Offline Angle
	Search}\label{sec:search_protocol}
$\Rea$ and $\Reb$ evaluate the same frozen-reference energy and differ only
in how the angles are searched: $\Rea$ treats all $2 \,p$ angles as free
variables at every depth, whereas $\Reb$ does so only while $p$ is small and
then restricts the search to the four amplitudes and exponents of the
schedule family below, so its search dimension stays fixed as $p$ grows.
Training and reported validation use disjoint frozen ensembles. Depth
$\ell$ angles seed depth $\ell+1$ by last-layer repetition $\mathsf{Rep}$ or
midpoint interpolation $\mathsf{Int}$:
\(\mathsf{Rep}_{\ell\to\ell+1}(x^{(\ell)})
=(x_1^{(\ell)},\ldots,x_\ell^{(\ell)},x_\ell^{(\ell)})\),
\([\mathsf{Int}_{\ell\to\ell+1}(x^{(\ell)})]_k
=\widetilde x_\ell\big((k-\tfrac12)/(\ell+1)\big)\),
\(k=1,\ldots,\ell+1\), for $x^{(\ell)}\in\{\gamma^{(\ell)},\beta^{(\ell)}\}$,
with $\widetilde x_\ell$ linearly interpolating the layer midpoints
$(j-\tfrac12)/\ell$, held constant beyond the first and last of them, and
constant at $\ell=1$. Every depth tries both, plus
fixed ramps, and keeps a refinement only on strict objective decrease. At
deep levels $\Reb$ projects the interpolant onto the four-parameter family
$\gamma_j=g_Ax_j^{g_K}$, $\beta_j=-b_A(1-x_j)^{b_K}$ at midpoints
$x_j=(j-\tfrac12)/\ell$, searching only $g_A,b_A,g_K,b_K$. Reference sizes
stay within exact-statevector reach; angles deploy at target depth and size.
The same  rules drive the $O_E$ and $O_s$
searches of Sections~\ref{sec:r1obj} and~\ref{sec:r5obj}.

\subsection{The \(O_E\) Objective and Correlated Energy Surrogate}
\label{sec:r1obj}
Let \(\gamma=(\gamma_1,\ldots,\gamma_p)\) and
\(\beta=(\beta_1,\ldots,\beta_p)\). For fixed \(J,h\),
\(\langle\cdot\rangle\) is the  state expectation of
Section~\ref{sec:qaoa_sk}, while \(\E_{J,h}\) averages the law induced by
\(H,\widetilde n,s\), as in BGMZ~\cite[Eq.~(3.1)]{bgmz2022constant} and
\cite[Eq.~(15)]{gulbahar2024mimo}. The physical energy density is as follows:
\begin{equation}
 V_p(\gamma,\beta)\triangleq \lim_{q\to\infty}
 \E_{J,h}\big[\langle\bar C_q\rangle/q\big]
 \label{eq:Vp_def}
\end{equation}
(Existence of this limit for the correlated M-QAM ensemble is not
asserted.) BGMZ calculate \eqref{eq:Vp_def} for independent disorder (their
Assumption~1). Replacing their covariance exponent by
$(\sigma_J^2,\sigma_h^2,A,\Bsf)$ of
\eqref{eq:scalar_variances}--\eqref{eq:scalar_B} defines a
covariance-matched, source-free, post-gauge Gaussian action surrogate,
evaluated at its selected stationary value, rather than a
controlled approximation to \eqref{eq:Vp_def}.
The construction (i) Gaussianizes the
Wishart Gram and drops its higher cumulants; (ii) drops deterministic means
and the planted reference; (iii) removes the middle-word marks \(m_i\) (Appendix~\ref{app:action}),
justified for BGMZ's sign-symmetric i.i.d.\ disorder but not here; and
(iv) projects the species, one per bit weight and axis, onto the
exchangeable averages \eqref{eq:scalar_variances}--\eqref{eq:scalar_B}. Equality
with \eqref{eq:Vp_def} is not asserted. The following lemma computes
\(V_p^{\rm corr}\), the model's stand-in for \(V_p\):
\(\Fcal^{(0)}_{A,\Bsf,0}\) denotes the free energy of the Gaussian model
(its stationary action value, \eqref{eq:free_energy}); source
differentiation gives its mean energy: couple an auxiliary \(t\) to the
energy and differentiate at \(t=0\); \(-\ii\) fixes the phase convention.
\begin{lemma}[Correlated surrogate readout of $V_p$]\label{lem:readout}
With paths \eqref{eq:path_weights},
contractions \eqref{eq:corr_contractions}, covariance exponents
\eqref{eq:ABD_action}--\eqref{eq:P0}, action \eqref{eq:free_energy},
and branch \(W^\star\) from
\eqref{eq:fixed_point}--\eqref{eq:corr_sce}, assumed to exist as a
differentiable continuation $W^\star(t)$ near $t=0$ with the denominator of
\eqref{eq:fixed_point} nonzero there, the covariance-matched readout is as follows with  all sums over \(a,b\in\Acal_p =\{\pm1\}^{2p}\):
\begin{align}
	&V_p^{\rm corr}(\gamma,\beta)
	=-\ii\,\partial_t\,\Fcal^{(0)}_{A,\Bsf,0}(t)\big|_{t=0}
	\nonumber\\
	&=-\ii\Big[(2\sigma_J^2 +A)\textstyle\sum_aW^\star_aK_a(W^\star)
	+(\sigma_h^2+\Bsf)\bar\Phi(W^\star)\Big]
	\label{eq:Vpcorr}
\end{align}
\end{lemma}
\begin{proof}
Appendix~\ref{app:vp_proofs}.
\end{proof}

Algorithm~\ref{alg:corrvp} evaluates Lemma~\ref{lem:readout} by continuation
and matrix-free Newton--GMRES (generalized minimal
residual)~\cite{saad1986gmres}. Its Newton--Krylov (NK) bundle is
\(\Pi_{\rm NK}=(\mathcal T,\epsilon_F,\epsilon_K,N_{\rm it},\Pi_K,\Lambda_{\rm d})\):
the continuation grid \(\mathcal T\subset[0,1]\), the residual and Krylov
tolerances \(\epsilon_F,\epsilon_K\), the Newton cap \(N_{\rm it}\), the GMRES restart controls \(\Pi_K\), and the
descending damping candidates \(\Lambda_{\rm d}\). Continuation switches on the
independent and correlated couplings in turn as follows:
\((\sigma_J^2,\sigma_h^2,A,\Bsf)\mapsto(u \, \sigma_J^2,u \, \sigma_h^2,v \, A,v \, \Bsf)\),
carrying the self-consistency equation (SCE) \eqref{eq:corr_sce} to the
reparameterized map \(\mathsf T_{u,v}\) on
path weights \(W\) (\(\sum_aW_a=1\)) and its residual
\(\mathsf F_{u,v}(W)\triangleq W-\mathsf T_{u,v}(W)\), whose
Jacobian is applied matrix-free as \(\mathsf F'_{u,v}\); the path
\((u,v)=(\zeta,0)\) then \((1,\zeta)\) ends at \((1,1)\), where
\(\mathsf F_{1,1}\) is the correlated equation itself. Steps~1-3 build their
objects from Proposition~\ref{prop:moments} and Appendix~\ref{app:action}: moments
\eqref{eq:scalar_moment_limits}, paths
\eqref{eq:path_weights},  \(K,\bar\Phi\)
\eqref{eq:corr_contractions}, and the Jacobian--vector product
\eqref{eq:jvp}.
\begin{algorithm}[t!]
\footnotesize
\caption{Energy objective $O_E$:  
$V_p^{\rm corr}(\gamma,\beta)$}
\label{alg:corrvp}
\begin{algorithmic}[1]
\Require $p,\gamma[1{:}p],\beta[1{:}p]$;
$\Xi_E=(b_R,b_I,\rho,\sigma_H^2,\mathrm{SNR}_{\rm dB},\mathrm{mode})$;
$\Pi_{\rm NK}=(\mathcal T,\epsilon_F,\epsilon_K,N_{\rm it},\Pi_K,\Lambda_{\rm d})$
\Ensure $\big(V_p^{\rm corr}(\gamma,\beta),
\|\mathsf F_{1,1}(W^\star)\|_\infty\big)$ 
\State Set $b\gets b_R+b_I$, $w_{P,r}\gets2^{b_P-r}$, and
$E_s\gets\sum_{P,r}w_{P,r}^2$
\State Set $\mathrm{SNR}_{\rm lin}\gets10^{\mathrm{SNR}_{\rm dB}/10}$;
set $\nu_{\rm n}\gets0$ if $\mathrm{mode}=\mathrm{ps}$, else
$\nu_{\rm n}\gets\sigma_H^2E_s/\mathrm{SNR}_{\rm lin}$
(the limit $\lim\sigma_n^2/n_t$ in each mode)
\State Obtain $(\sigma_J^2,\sigma_h^2,A,\Bsf,D{=}0)$ by
\eqref{eq:scalar_moment_limits}
\State Build $\Acal_p,(\Phi_a,Q_\beta(a))$ by
\eqref{eq:path_weights}, normalize $Q_\beta$ so $\sum_aQ_\beta(a)=1$, and encode
$a\circ b$ as XOR; implement $K,\bar\Phi,\mathsf F'_{u,v}$ by FWHT
\State Define normalized $\mathsf T_{u,v}$ from \eqref{eq:corr_sce} under
$(\sigma_J^2,\sigma_h^2,A,\Bsf)\mapsto
(u\sigma_J^2,u\sigma_h^2,vA,v\Bsf)$; set
$\mathsf F_{u,v}(W)\gets W-\mathsf T_{u,v}(W)$, $W\gets Q_\beta$
\ForAll{$(u,v)=(\zeta,0)$ then $(1,\zeta)$,
$\zeta\in\mathcal T\setminus\{0\}$ in increasing order}
  \For{$j\le N_{\rm it}$ while $\|\mathsf F_{u,v}(W)\|_\infty>\epsilon_F$}
    \State Solve $\delta\gets\mathsf{GMRES}(\mathsf F'_{u,v}(W),r;\Pi_K,
    \epsilon_K)$ with $r\gets\mathsf F_{u,v}(W)$
\State Update $W\gets W-\lambda\delta$, first $\lambda\in\Lambda_{\rm d}$
    strictly decreasing $\|\mathsf F_{u,v}\|_\infty$
  \EndFor
  \State Record $W_{u,v}\gets W$ once
  $\|\mathsf F_{u,v}(W)\|_\infty\le\epsilon_F$
\EndFor
\State Set $W^\star\gets W$, $K_a^\star\gets K_a(W^\star)$,
$\bar\Phi^\star\gets\bar\Phi(W^\star)$
\State \Return $\big(-\ii[(2\sigma_J^2+A)\sum_aW_a^\star
K_a^\star+(\sigma_h^2+\Bsf)\bar\Phi^\star],\;
\|\mathsf F_{1,1}(W^\star)\|_\infty\big)$
\end{algorithmic}
\end{algorithm}
Each $a\in\Acal_p$ is a one-qubit ket-bra path; suffixes recover slice
signs and $a\circ b$ is XOR. The contractions are as follows:
\begin{equation}
\begin{aligned}
K_a(W)&\triangleq \sum_{b\in\Acal_p}W_b\Phi_{a\circ b},&
\bar\Phi(W)&\triangleq \sum_{a\in\Acal_p}W_a\Phi_a 
\end{aligned}
\label{eq:corr_contractions}
\end{equation}
The continuation $Q_\beta=W^{\rm ind}(0)\to W^{\rm ind}(1)=W^{\rm corr}(0)\to W^\star$ activates $(\sigma_J^2,\sigma_h^2)$ before $(A,\Bsf)$; $u,v$ are
homotopy coordinates (superscripts label the two legs) and uniqueness is
not asserted. $\Pi_{\rm NK}$ archives
solver choices and residuals, while \eqref{eq:jvp} gives matrix-free
Jacobian--vector products. The formal source and envelope theorem yield
\eqref{eq:Vpcorr}; with $c_1^2=\sigma_h^2$, $c_2^2=2\sigma_J^2$,
$g_1(\lambda)=g_2(\lambda)=-\lambda^2/2$, and the ket/bra branches
swapped, at $A=\Bsf=0$ it reduces to
BGMZ~\cite[Eq.~(3.10)]{bgmz2022constant}, otherwise the stated
covariance-matched surrogate.  

Table~\ref{tab:complexity} separates offline and deployment costs, with
\(N_{\rm K}\) Krylov iterations per Newton step, \(N_{\rm it}\) Newton
steps per continuation node, \(C_{\mathcal T}=2(|\mathcal T|-1)\) nodes per
objective evaluation, and \(N_{\rm fev}\) evaluations per angle search.
FWHT evaluation of \eqref{eq:corr_sce} and \eqref{eq:jvp} costs $O(p \, 4^p)$
time/memory and Newton costs $O(N_{\rm K} \,p \, 4^p)$ without a stored Jacobian;
direct sums cost $O(4^{2  p})$, which bounds the analytic depth. 

\begin{table}[t]
	\centering
	\caption{Deployment costs for objective \(O_E\) }
	\label{tab:complexity}
	\begingroup
	\footnotesize
	\setlength{\tabcolsep}{0.8pt}
	\renewcommand{\arraystretch}{1.05}
	\begin{tabularx}{\columnwidth}{|>{\raggedright\arraybackslash}p{0.168\columnwidth}|X|
			>{\raggedright\arraybackslash}p{0.212\columnwidth}|X|}
		\hline
		\multicolumn{2}{|c|}{\emph{Offline design (classical)}} &
		\multicolumn{2}{c|}{\emph{Deployment (quantum), per shot}}\\
		\hline
		Stage & Cost & Stage & Cost\\
		\hline
		Paths/phases & \(O(p\, 4^p)\) setup &
		Total gates & \(O(p \, q^2)\)\\
		\hline
		SCE map & \(O(p \,4^p)\) &
		Circuit depth & \(O(p \, q)\) layers\\
		\hline
		Newton step & \(O(N_{\rm K}p \, 4^p)\) &
		Sampling & \(q\) readouts; \(S \, q\) bits\\
		\hline
		Angle search &
		\(O(N_{\rm fev}C_{\mathcal T}\)\newline\(N_{\rm it}N_{\rm K}p \,4^p)\) &
		Classical pre/post & \(J,h\): \(O(n_r q^2)\);\newline
		\(S\)-best: \(O(Sq^2)\)\\
		\hline
	\end{tabularx}
	\endgroup
\end{table}
\emph{Auxiliary finite-\(q\) oracles.} Closed-form \(p=1\), a
column-conditioned cavity estimate (leave one channel column out), and
subset-truncated \(p=2\) (retaining expansion subsets \(\mathcal N\),
\(|\mathcal N|\le1\)) validate the implementation without redefining
\(V_p\); all three are implemented in the reproducibility
capsule~\cite{repository}, and only the \(p=1\) circuit identities are exact.

\subsection{The \(O_s\) Objective and General-Depth Sampling Law}
\label{sec:r5obj}

Algorithm~\ref{alg:corrvp} targets the energy surrogate \eqref{eq:Vpcorr},
whereas best-of-\(S\) decoding depends on ML mass. Reported \(O_s\) packs
therefore minimize frozen-reference \(\rs\) at every prescribed
depth. In parallel, \eqref{eq:planted_pp} gives an exact planted ensemble law
for  fixed \(p\), an independent check on the sampling engine rather than
a design objective, whose annealed rate, the log of the mean mass (quenched
and annealed averages in the sense of~\cite{tanaka2002}), is as follows:
\begin{equation}
	\rpl(p)\triangleq-q^{-1}\log\E P^{\rm pl}_{q,p}
	\label{eq:rpl_def}
\end{equation}
A related success-probability exponent appears in~\cite{boulebnane2024sat}.
Algorithm~\ref{alg:sampling} evaluates \eqref{eq:rs_def} and
\eqref{eq:rpl_def} on separate branches: the reference-rate branch prepares
each frozen instance's depth-\(p\) state exactly and averages
\(\log P^{\rm ML}\), while the exact-ensemble branch convolves the
single-symbol table.  
The depth-one formulas are presented first only because their single-symbol
table is four-dimensional there and can be convolved in practice.  Define
\(\gamma_\ell^{\rm eff}\triangleq \gamma_\ell/q\) from \eqref{eq:phase_q}.

\subsubsection{Depth-one construction}
For the per-shot probability
\(P^{\rm pl}_{q,p}=|\langle s|\gamma,\beta\rangle|^2\) of measuring   \(s\), the site-wise gauge \(\tau_i=z_is_i\) factorizes the depth-one mixer element
\(m(z)=\langle z|U_B(\beta)|s\rangle\), \(U_B(\beta)=e^{-\ii\beta B}\) with
\(B=\sum_iX_i\), as follows:
\begin{equation}
 m(z)=\prod_{i=1}^{q}g(\tau_i); \,
 g(+1)=\cos\beta; \,  g(-1)=-\ii\sin\beta
 \label{eq:planted_mix}
\end{equation}
and let the symbol-domain errors
\(\varepsilon=T(s-z),\varepsilon'=T(s-z')\in\mathbb R^{2n_t}\) carry the ket
and bra words against \(s\) through the bit-to-symbol map \(T\)
of~\eqref{eq:T_matrix}; these are the depth-one case of
\(\varepsilon^{(\ell)}\) below.
Identify their R/I
halves as
\(\varepsilon^{\rm c}=\varepsilon_R+\ii\varepsilon_I\in\mathbb C^{n_t}\).
For \(\langle x,y\rangle\triangleq
\sum_ux^{\rm c}_u\overline{y^{\rm c}_u}\), define the error Gram
\(\Sigma\) and   phase-noise kernel \(D_1\) (\(G_p\) and \(D_p\) defined below) as follows:
\begin{equation}
 \begin{aligned}
 \Sigma = G_1 &=\begin{pmatrix}|\varepsilon^{\rm c}|^2 &
 \langle\varepsilon,\varepsilon'\rangle\\
 \overline{\langle\varepsilon,\varepsilon'\rangle} &
 |\varepsilon'^{\rm c}|^2\end{pmatrix} \\[-1pt]
 D_1&=\ii\gamma_1^{\rm eff}\begin{pmatrix}1&0\\0&-1\end{pmatrix}
 +\sigma_n^2(\gamma_1^{\rm eff})^2
 \begin{pmatrix}1&-1\\-1&1\end{pmatrix} 
 \end{aligned}
\label{eq:planted_gram}
\end{equation}
Equations~\eqref{eq:planted_gram}--\eqref{eq:planted_pp} set
\(\sigma_H^2=1\); for general channel variance replace \(\Sigma\) (or
\(G_p\)) by \(\sigma_H^2\Sigma\) (or \(\sigma_H^2G_p\)), as
Algorithm~\ref{alg:sampling} does. Here \(G_p,D_p,\Sigma\) are sampling-side error-Gram and kernel objects,
distinct from the detector \(G\), the scalar \(D\), and \(\Sigma_q\).

\begin{proposition}[exact depth-one specialization]\label{prop:planted_p1}
With \eqref{eq:planted_mix}--\eqref{eq:planted_gram} and per-symbol bit
weights \(w_{P,r}=2^{b_P-r}\),
\begin{equation}
 \E P^{\rm pl}_{q,1}=2^{-2q}\!\!\sum_{s,z,z'}\!m(z)m(z')^{*}
 \big[\det(I_2+\Sigma D_1)\big]^{-n_r} 
 \label{eq:planted_p1}
\end{equation}
while the determinant kernel depends on \((s,z,z')\) only through  
\begin{equation}
 \iota=(d,d',u,v)=\big(|\varepsilon^{\rm c}|^2,\ |\varepsilon'^{\rm c}|^2,\
 \Re\langle\varepsilon,\varepsilon'\rangle,\
 \Im\langle\varepsilon,\varepsilon'\rangle\big)
 \label{eq:planted_inv}
\end{equation}
each component of which is a sum of per-symbol contributions. Mixer weights depend on the
full bit pattern, so they are collision-aggregated within each invariant cell
before the common kernel is applied. Then, \eqref{eq:planted_p1} equals the \(n_t\)-fold convolution over
\(\iota\) of the per-symbol weights \eqref{eq:planted_table}, without
enumerating \(2^{q}\) configurations.
\end{proposition}
\begin{proof}
Appendix~\ref{app:vp_proofs}.
\end{proof}

At general depth, the amplitude replicates into \(p\) forward and \(p\)
backward configurations, so the depth-one construction repeats with \(2p\)
charges
\(c_p=(\gamma_1^{\rm eff},\ldots,\gamma_p^{\rm eff},
-\gamma_1^{\rm eff},\ldots,-\gamma_p^{\rm eff})\) and phase--noise kernel
\(D_p=\ii\diag(c_p)+\sigma_n^2c_pc_p^{\!\top}\) as follows:
\begin{equation}
 \E P^{\rm pl}_{q,p}=\E_{\boldsymbol s}\Big[2^{-q}\!\!\sum_{\mathcal P_p}\!
 W_p(\mathcal P_p)\big[\det(I_{2p}+D_pG_p)\big]^{-n_r}\Big]
 \label{eq:planted_pp}
\end{equation}
where \(\mathcal P_p=(z^{(1)},\ldots,z^{(p)},z'^{(1)},\ldots,z'^{(p)})\)
lists the \(p\) ket and \(p\) bra slice words,
\(\varepsilon^{(\ell)}=T(s-z^{(\ell)})\) and
\(\varepsilon^{(p+\ell)}=T(s-z'^{(\ell)})\) are their symbol-domain errors in the
complex form of \eqref{eq:planted_inv}, and with transition signs
\(\zeta^{(\ell)}_i=z^{(\ell+1)}_iz^{(\ell)}_i\) (\(\ell<p\)),
\(\zeta^{(p)}_i=s_iz^{(p)}_i\), and \(\zeta'\) likewise from \(z'\), the error
Gram and path weight are  
$(G_p)_{\ell m}=\langle\varepsilon^{(\ell)}\!,\varepsilon^{(m)}\rangle$ and
 $W_p (\mathcal P_p)=\prod_{i=1}^{q}\prod_{\ell=1}^{p}
 g_{\beta_\ell}(\zeta^{(\ell)}_i)
 g_{\beta_\ell}(\zeta'^{(\ell)}_i)^{*}$ 
with \(g_\beta\) the mixer element of \eqref{eq:planted_mix}. Hermitian
\(G_p\) has \(4p^2\) real entries additive across symbols, and \(W_p\) is
multiplicative across them, so an exact convolution exists at every fixed
\(p\) but costs \(q^{\Theta(p^2)}\); practicality singles out \(p=1\), where
\(G_p=\Sigma\) and \(W_p=m(z)m(z')^{*}\).

\begin{algorithm}[t!]
\footnotesize
\caption{Sampling objective $O_s=\rs(\theta)$ and exact planted
benchmark $\rpl(p)$}
\label{alg:sampling}
\begin{algorithmic}[1]
\Require $p,\gamma[1{:}p],\beta[1{:}p]$;
$\Xi_s=(b_R,b_I,n_t,n_r,\sigma_H^2,\mathrm{SNR}_{\rm dB},\mathrm{mode})$;
$\mathrm{target}\in\{\text{reference rate},\text{exact ensemble}\}$; frozen
$\mathcal I$ for the ref. rate
\Ensure $\rs(\theta;\mathcal I)$, or $\rpl(p)$, per $\mathrm{target}$
\State Set $b,q,w_{P,r},E_s$ and
$\sigma_n^2=(n_t\sigma_H^2)^{\mathbf1_{\{\mathrm{mode=rx}\}}}E_s10^{-\mathrm{SNR}_{\rm dB}/10}$
\State Form $c_p=(\gamma^{\rm eff},-\gamma^{\rm eff})$,
$D_p=\ii\diag(c_p)+\sigma_n^2c_pc_p^{\!\top}$, the $2p$ errors
$\varepsilon^{(\ell)}$, $W_p(\mathcal P_p)$, and additive $G_p$ in
\eqref{eq:planted_pp}
\If{$\mathrm{target}=\text{reference rate}$}
  \State Set $q\gets q_{\rm ref}$ from $\mathcal I$
  \ForAll{$I\in\mathcal I$}
    \State Find $\widehat z_I=\arg\min_z\bar C_{q,I}(z)$, ties broken
    lexicographically
    \State Prepare $|\psi_I\rangle=\prod_{\ell=p}^{1}e^{-\ii\beta_\ell B}
    e^{-\ii\gamma_\ell\bar C_{q,I}}|+\rangle^{\otimes q}$ 
    \State Set $P_I=|\langle\widehat z_I|\psi_I\rangle|^2$
  \EndFor
  \State \Return $\rs(\theta;\mathcal I)=
  -(q_{\rm ref}|\mathcal I|)^{-1}\sum_{I\in\mathcal I}\log P_I$
\EndIf
\If{$\mathrm{target}=\text{exact ensemble}$}
  \State Enumerate the $2^{b(2p+1)}$ one-symbol tuples into the
  $4p^2$-coordinate table $A_p$, weighted by $2^{-b}$ times the one-symbol
  factor of $W_p$ (at $p=1$, \eqref{eq:planted_table}); for $p=1$ use
  \eqref{eq:planted_inv}-\eqref{eq:planted_mask}
  \State Convolve $C_p\gets\operatorname{DFT}^{-1}[(\operatorname{DFT}A_p)^{n_t}]$
  \State Set $\mathcal U_p\gets\{G:C_p[G]\ne0,\ G\text{ reachable}\}$
  \State Accumulate $\Pi\gets2^{-q}\sum_{G\in\mathcal U_p}
  C_p[G]\det(I_{2p}+\sigma_H^2D_pG)^{-n_r}$
  \State \Return $\rpl(p)=-q^{-1}\log\Pi$
\EndIf
\end{algorithmic}
\end{algorithm}

\subsubsection{Evaluation of the planted law}
Three ingredients make \eqref{eq:planted_p1} computable without enumerating
\(2^{q}\) configurations: a per-symbol weight table, a reachability mask,
and an alias-free transform.
For one symbol, assign to \((s,\tau,\tau')\) the invariant \(\iota\) of
\eqref{eq:planted_inv} and the following weight:
\begin{equation}
 \varpi(s,\tau,\tau')=2^{-b}\prod_{i=1}^{b}g(\tau_i)\,g(\tau'_i)^{*}
 \label{eq:planted_table}
\end{equation}
the per-symbol factor of \(m(z)m(z')^{*}\) in \eqref{eq:planted_p1}, with
\(2^{-b}\) its share of the uniform average over the planted symbol.
Only reachable cells contribute, since \(\Sigma\succeq0\) forces the following:
\begin{equation}
 \det\Sigma\ge0,\qquad \|\varepsilon-\varepsilon'\|^2
 =d+d'-2u\ \ge0 
 \label{eq:planted_mask}
\end{equation}
and $\Re\det(I_2+\Sigma D_1)=1+\sigma_n^2(\gamma^{\rm eff})^2
(d{+}d'{-}2u)+(\gamma^{\rm eff})^2\det\Sigma\ge1$ there, so the inverse
determinant stays bounded and the transform is pole-free. The mask is what makes the exact law
computable: with the nonzero-support condition it leaves a few percent of the
box, so the \(O(p^3)\) determinant is evaluated on
\(\mathcal U_p\ll\mathcal L_p\) cells. Because \(\iota\) is additive across
symbols, the \(n_t\)-symbol law is the \(n_t\)-fold convolution
\(C_p=\operatorname{DFT}^{-1}[(\operatorname{DFT}A_p)^{n_t}]\) of the one-symbol table
\(A_p\) of \(\varpi\) values, with \(\operatorname{DFT}\) the alias-free discrete
Fourier transform on the \(4p^2\)-axis invariant lattice \(\mathcal L_p\), four-dimensional in
\(\iota\) at \(p=1\). All four invariants are multiples of
\(4\), so with \(d_{\max}=\{[2(2^{b_R}\!-\!1)]^2+[2(2^{b_I}\!-\!1)]^2\}/4\)
the table is held on an
\((n_td_{\max}\!+\!1)^2\times(2n_td_{\max}\!+\!1)^2\) array by the encoder
\(\mathrm{enc}(\iota)=(d/4,d'/4,u/4+d_{\max},v/4+d_{\max})\); the $n_t$-fold
convolution accumulates the offset to $n_td_{\max}$, and the output is
decoded by \(\mathrm{dec}(i,j,k,l)=4(i,j,k-n_td_{\max},l-n_td_{\max})\); only cells with
\(\chi(\iota)=\mathbf 1\{dd'-u^2-v^2\ge0,\ d+d'-2u\ge0\}\) survive.
Algorithm~\ref{alg:sampling} implements these in its exact-ensemble
branch, returning the
frozen-reference rate \(\rs\) at any tractable depth, or the exact
four-coordinate ensemble law at \(p=1\).
Equation~\eqref{eq:planted_pp} stays exact for every fixed \(p\), but its
lattice grows as \(q^{\Theta(p^2)}\) at fixed bit depth, bounded by
\((n_td_{\max}{+}1)^{2p}(2n_td_{\max}{+}1)^{4p^2-2p}\); the proxy
\(q^{4p^2}\) reaches \(10^{29}\) at
\(q=64,p=2\), so reported \(O_s\) packs use the frozen-reference objective
beyond $p = 1$. Table~\ref{tab:planted_cost} lists the design and benchmark costs; both
objectives set angle values in the same circuit, so deployment is the
per-shot cost of Table~\ref{tab:complexity} in either case.
Either evaluator can serve the design, and both are offline: incurred once per
pack and amortized over every deployment instance and size. On hardware, transmitting a known word, the
planted rate $\hat r_{\rm pl}$ is estimated from shots at the deployed size  (its
gap to the ML rate is bounded in Sec.~\ref{sec:rate_relations}), so no reference size is
needed; the packs reported here are instead evaluated by exact statevector
simulation, whose \(\Theta(N_{\rm tr}pq2^{q})\) cost per call caps the design
size at \(q_{\rm ref}\le24\) and makes transfer a practical necessity rather
than a property of the objective.

\begin{table}[t]
	\centering
	\caption{Costs for the sampling objective \(O_s\) and  ensemble law}
	\label{tab:planted_cost}
	\begingroup
	\footnotesize
	\setlength{\tabcolsep}{1.0pt}
	\renewcommand{\arraystretch}{1.05}
	\begin{tabularx}{\columnwidth}{|>{\raggedright\arraybackslash}p{0.200\columnwidth}|X|
			>{\raggedright\arraybackslash}p{0.180\columnwidth}|X|}
		\hline
		\multicolumn{2}{|c|}{\emph{Offline design, any \(p\)}} &
		\multicolumn{2}{c|}{\emph{Exact ensemble law, fixed \(p\)}}\\
		\hline
		Stage & Cost & Stage & Cost\\
		\hline
		Quantum eval. & \(\Theta(\epsilon^{-2}P^{-1}\)\newline\(\log(1/\delta))\) shots/inst. &
		Symbol table & \(O((p \,b{+}p^2)2^{b}4^{p b})\)\\
		\hline
		Classical eval. & \(\Theta(N_{\rm tr}pq2^{q})\)\newline per value, gradient &
		Invariant lattice & \(|\mathcal L_p|=q^{\Theta(p^2)}\)\\
		\hline
		Ensemble & \(O(N_{\rm tr})\) avg.\ of logs &
		Convolution & \(O(|\mathcal L_p|\log|\mathcal L_p|)\)\\
		\hline
		Search & \(N_{\rm fev}\) obj.\ calls &
		Determinants & \(O(|\mathcal U_p|p^3)\), \(\mathcal U_p\!\ll\!\mathcal L_p\)\\
		\hline
	\end{tabularx}
	\endgroup
\end{table}

\subsubsection{Estimator and rate relations}\label{sec:rate_relations}
Three rates appear, differing in which word is targeted and in the order of
expectation and logarithm. The design objective \eqref{eq:rs_def} is a
$q_{\rm ref}$ sample average of $-\log P^{\rm ML}$ over $\mathcal I$; at fixed
$\theta$ it estimates the population rate
$r_{\rm ML}(\theta,p)=-q^{-1}\E[\log P^{\rm ML}_{q,p}(\theta)]$, with sampling
error vanishing as $|\mathcal I|\to\infty$ when reference and deployment laws
match, although this says nothing about the error of selecting $\theta$.
All rates here are finite-$q$ objects; where
Proposition~\ref{prop:shot_bound} invokes a $q\to\infty$ limit $r_\infty(p)$, its existence
is an explicit hypothesis there.
Writing $\rml(p)\triangleq r_{\rm ML}(\theta^\star_p,p)$ for the
offline-selected pack $\theta^\star_p$ gives the quenched ML rate
($\E\log$) that the depth law and Proposition~\ref{prop:shot_bound} use;
the residual
reference/deployment discrepancy is transfer, measured in
Section~\ref{sec:r5sims} at 2\% on matched fits. The benchmark
\eqref{eq:rpl_def} instead reports $\rpl(p)$, annealed ($\log\E$) and planted
(on $s$), computable efficiently only at $p=1$; the annealed ML counterpart
$\hat r^{\rm ann}_{\rm ML}(p)=-q^{-1}\log\E P^{\rm ML}_{q,p}$ is what
Fig.~\ref{fig:objcmp}(b),(d) fit. Rate types must therefore be matched before
a difference is read as transfer. Since $P^{\rm ML}=P^{\rm pl}$ whenever
$\hat z=s$, the quenched rates $\rml$ and
$r^{\rm q}_{\rm pl}\triangleq-q^{-1}\E\log P^{\rm pl}$ differ only through
word errors: $|\rml-r^{\rm q}_{\rm pl}|
=q^{-1}\big|\E[\log(P^{\rm ML}/P^{\rm pl});\hat z\neq s]\big|
\le q^{-1}\sqrt{P_e\E[\log^2(P^{\rm ML}/P^{\rm pl})]}$ with
$P_e=\Pr[\hat z\neq s]$, and coincide over the tested 10--20\,dB
word-equality range.

\subsection{Shot Budgets and the Decoder Gap}\label{sec:shotbound}
Subexponential shot budget cannot move the leading large-$q$ sampling
exponent, so lowering the plateau's exponent requires improving \(\rml\) rather
than spending more shots; at finite $q$, larger $S$ still shrinks the
gap \eqref{eq:shot_gap}. For hit probability \(h_S(P)=1-(1-P)^S\), Bernoulli
gives \(P\le h_S(P)\le SP\).

\begin{proposition}[shot budget and sampling exponent]
\label{prop:shot_bound}
Let \(P^{\rm ML}_{q,p}\) be the per-shot ML mass and \(\rml(p)\) its
quenched rate; assume $P^{\rm ML}_{q,p}>0$ a.s.\ with
$\E\,|\log P^{\rm ML}_{q,p}|<\infty$. Then,
\begin{equation}
 \rml(p)-\frac{\log S}{q}\ \le\
 -\frac1q\E\big[\log h_S(P^{\rm ML}_{q,p})\big]\ \le\ \rml(p)
 \label{eq:shot_sandwich}
\end{equation}
so a subexponential budget leaves the exponent unchanged, while for
\(S_q=e^{\kappa q+o(q)}\) and \(-q^{-1}\log P^{\rm ML}\to r_\infty(p)\) in
probability (existence of this limit is a hypothesis) with
\(\{q^{-1}\log h_S\}\) uniformly integrable, the hit
exponent tends to \((r_\infty(p)-\kappa)_+\), with
\((x)_+\triangleq\max\{x,0\}\). 
Let \(\mathrm{BER}_{p,S}\) be the ensemble BER of minimum-cost selection among
\(S\) depth-\(p\) samples and \(\mathrm{BER}_{\rm ML}\) the exact-ML BER under
the same labeling and instance law, the ML word being almost surely unique
for continuous \(H,\widetilde n\).
Since minimum-cost selection returns the ML
word whenever it is drawn as follows:
\begin{equation}
 \big|\mathrm{BER}_{p,S}-\mathrm{BER}_{\rm ML}\big|\ \le\
 \E\big[(1-P^{\rm ML}_{q,p})^{S}\big]
 \label{eq:shot_gap}
\end{equation}
\end{proposition}
\begin{proof}
Take logs and expectations in \(P\le h_S(P)\le SP\) for
\eqref{eq:shot_sandwich}. For \eqref{eq:shot_gap}, the decoders agree on an ML hit; otherwise they differ
by at most one, an event of probability \((1-P^{\rm ML})^{S}\).
\end{proof}
Therefore, the residual BER gap is bounded by the probability of never
drawing the ML word in \(S\) shots, which shrinks with the budget \(S\)
and, where depth raises \(P^{\rm ML}\), with depth. 
Exact ML minimizes word error, so a single sampled
instance may carry fewer bit errors than the ML word; no from-above or
monotone approach to $\mathrm{BER}_{\rm ML}$ is guaranteed, and the reported
curves lie above their floors empirically.

\begin{table}[t]
	\centering
	\caption{Simulation Groups and Parameters}
	\label{tab:simulation_census}
	\begingroup
	\footnotesize
	\setlength{\tabcolsep}{0.7pt}
	\setlength{\arrayrulewidth}{0.3pt}
	\renewcommand{\arraystretch}{0.94}
	\begin{tabular}{|
			>{\centering\arraybackslash}p{0.115\columnwidth}|
			>{\centering\arraybackslash}p{0.220\columnwidth}|
			>{\centering\arraybackslash}p{0.220\columnwidth}|
			>{\centering\arraybackslash}p{0.085\columnwidth}|
			>{\centering\arraybackslash}p{0.160\columnwidth}|
			>{\centering\arraybackslash}p{0.140\columnwidth}|}
		\hline
		Fig. & $M$ & $n_t{\times}n_r$ & $p_{\max}$ & $S$ & $N$ \\
		\hline
		\ref{fig:objcmp}, \ref{fig:objective} & 16/64 & $5\times5$, $3\times3$ &
		6/10 & $2^{14}$ & 1000 \\
		\hline
		\ref{fig:paired_q24} & 16/64 & $6\times6$, $4\times4$ &
		8 & $2^{9,14,18}$ & 2000 \\
		\hline
		\ref{fig:depth_panels}(a, b) & 16/64 & $6\times6$, $4\times4$ &
		24 & $2^{9,14,18}$ & 2000 \\
		\hline
		\ref{fig:depth_panels}(c) & 128 & $2\times10$ &
		12 & $2^{10,14,18}$ & 8000 \\
		\hline
		\ref{fig:depth_panels}(e) & 1024 & $2\times48$ &
		16 & $2^{10,14,18}$ & 5000 \\
		\hline
		\ref{fig:depth_panels}(d, f) & 256/1024 & $2\times2$ &
		24 & $2^{9,14,18}$ & 2k/tier \\
		\hline
		\ref{fig:depth_panels}(g) & 4096 & $2\times2$ &
		30 & $2^{10,14,18}$ & 1500 \\
		\hline
		\ref{fig:depth_panels}(h) & \makecell[c]{16/64, 256/1024} &
		$2{\times}2$--$5{\times}5$ &
		20 & -- & 0.4/2k \\
		\hline
		\ref{fig:corr_depth_size}(a--e) & 16 & $5\times5$ &
		8 & -- & 24/pt \\
		\hline
		\ref{fig:largen} & 16 & $2{\times}2$--$128{\times}128$ &
		2 & -- & 0.3--8k/pt \\
		\hline
	\end{tabular}
	\endgroup
\end{table}
\section{Simulations}\label{sec:numerics}

Statevectors evaluate tractable deployments; cavity and subset
calculations are size-scaling diagnostics only. Table~\ref{tab:simulation_census}
catalogs the principal \(\Rea/\Reb\), scheduled \(O_s\), and paired
simulations; the captions of
Figs.~\ref{fig:objcmp} and \ref{fig:objective} specify the \(O_E/O_s\)
settings. ``pt'' means per point, a dash marks an
expectation, and slashes follow listed order. Row~(h) lists five fit-only
simulations as sets; its other two fit tiers are rows~(a),(g).

\subsection{Comparison of \(O_E\) and \(O_s\)}\label{sec:objective}

\subsubsection{Identical-Instance \(O_E\) versus \(O_s\)}
\label{sec:oe_vs_os}
Objectives $O_E$ and $O_s$ minimize \eqref{eq:Vpcorr} and the
frozen-reference ML rate~\eqref{eq:rs_def}. As shown in Fig.~\ref{fig:objcmp}(a) for 16-QAM $5\times5$ at 20\,dB,
disjoint from design, and the common budget $S=2^{14}$, the two routes
nearly coincide at $p=1$, and the $O_s$ advantage then grows with depth,
its BER falling more than four orders of magnitude below $O_E$ by $p=6$.  The $3\times3$ 64-QAM records
at 28\,dB in (c) repeat this pattern against the exact-ML line, the
matched-depth gap reaching $36\times$ and the energy route sitting $89\times$
above exact ML. The bound
\eqref{eq:shot_gap} is satisfied at every measured point in these paired
records.

16-QAM fits in Fig.~\ref{fig:objcmp}(b) give the following:
\begin{equation}
 \hat r^{\rm ann}_{O_s}=0.503\,p^{-0.451},\qquad
 \hat r^{\rm ann}_{O_E}=0.494\,p^{-0.162}
 \label{eq:objective_exponents}
\end{equation}
while the $3\times3$ 64-QAM fits in (d) give $0.335$ and $0.073$. The
prefactors nearly coincide while the exponents split: under the fitted
laws, doubling the depth lowers the $O_E$ rate by only about $11\%$
($5\%$ at 64-QAM), against $27\%$ ($21\%$) for $O_s$, so the deficit
is in the depth response, not the depth-one start.  
 
Fig.~\ref{fig:objective} shows  SNR dependence of the paired
designs: Figs.~\ref{fig:objective}(a)--(d) for varying per-symbol SNR at each depth for $O_E$ and
then $O_s$, on 16-QAM ($5\times5$) over 0--20\,dB and  64-QAM ($3\times3$)
over 8--28\,dB while Figs.~\ref{fig:objective}(e)--(h) show both
energy references, $\langle\bar C_q\rangle/C_{\rm ML}$ and
$\langle\bar C_q\rangle/C_{\rm tx}$, for varying $p$ at 12\,dB (16-QAM) and
20\,dB (64-QAM), in the same order;
$\langle\bar C_q\rangle/C_{\rm tx}$ agrees to $10^{-4}$ throughout, since the
transmitted word is the ML word over this range.
$O_s$ curve for $p = 10$  in (d) stays within $1.33\times$ ML over the full
8--28\,dB span, whereas $O_E$ in (b) about $57\times/89\times$ ML
at 24/28\,dB. Consistent with the small $O_E$ exponents of
\eqref{eq:objective_exponents}, the $O_E$ depth ladders in (a),(b) remain
tightly bunched across the sweep, while the $O_s$ ladders in (c),(d)
spread toward the exact-ML reference. $O_s$ energy ratios in
(g),(h) are nonmonotone in depth although its BER falls, so  $O_s$ trades mean energy for ML mass.

\begin{figure}[t]
\centering
\includegraphics[width=\columnwidth]{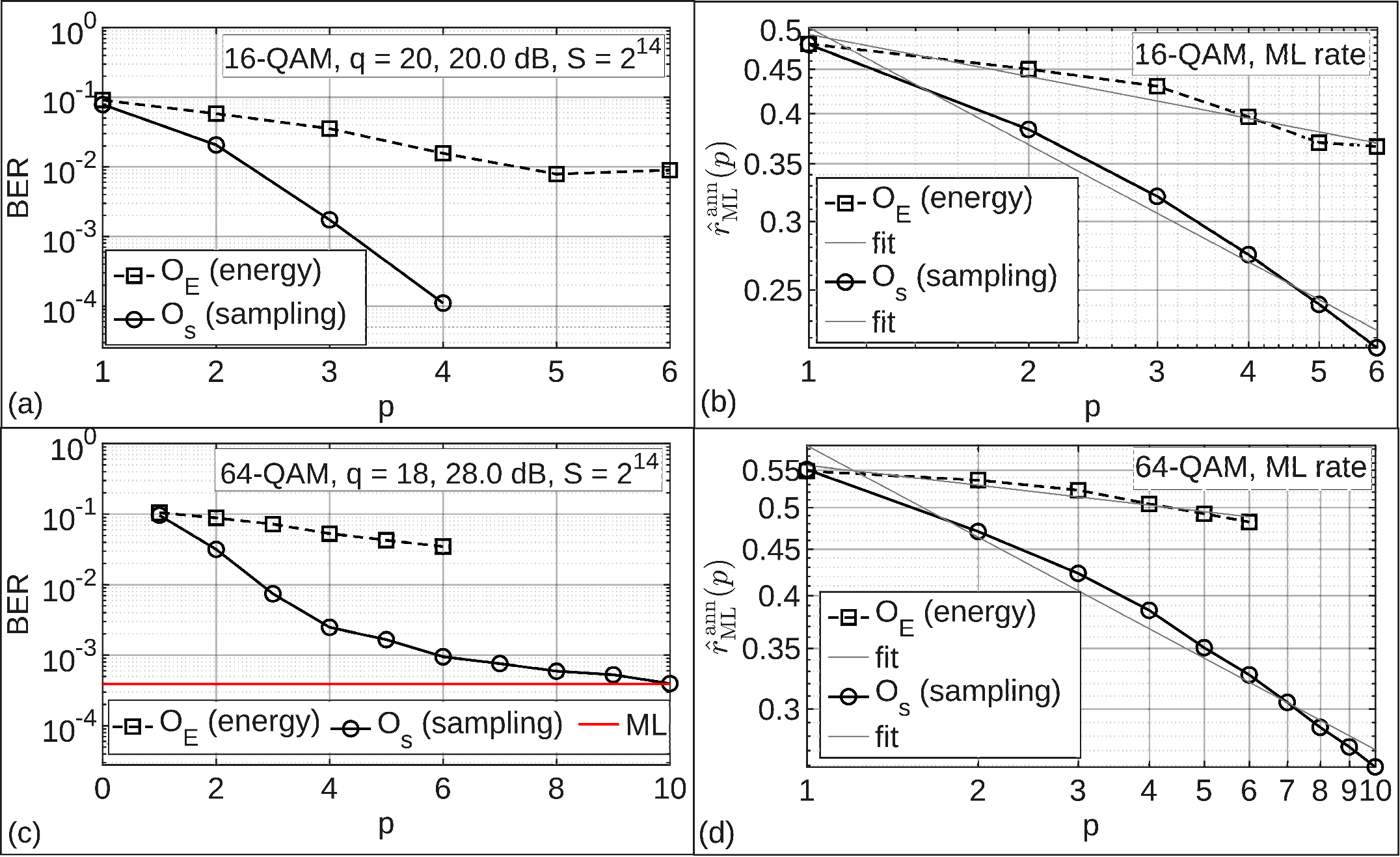}
\caption{Identical-instance $O_E$  versus $O_s$ at
$S=2^{14}$. (a) exact best-of-$S$ BER, 16-QAM ($5\times5$) at 20\,dB;
(b) annealed rate $\hat r^{\rm ann}_{\rm ML}(p)$ for the same runs, with
power-law fits; (c), (d) repeat (a), (b) for 64-QAM ($3\times3$) at 28\,dB, where
the flat line in (c) is the exact-ML BER.}
\label{fig:objcmp}
\end{figure}

\begin{figure}[t]
\centering
\includegraphics[width=\columnwidth]{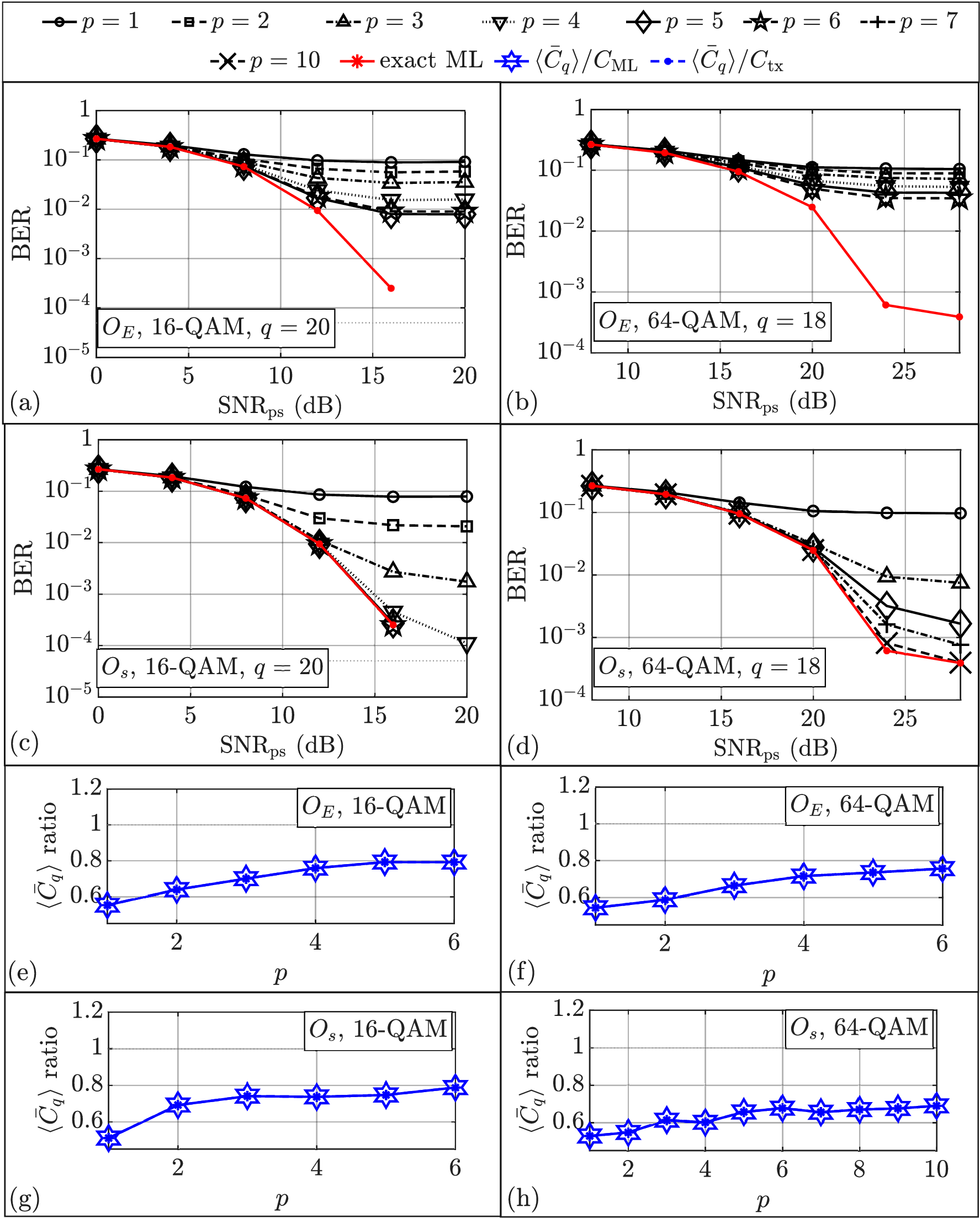}
\caption{Paired $O_E$ and $O_s$ results at $S=2^{14}$.  BER vs. $\mathrm{SNR}_{\rm ps}$ by depth:
(a) $O_E$ at 16-QAM ($5\times5$), (b) $O_E$ at    64-QAM ($3\times3$), (c) and (d) for $O_s$ at the same two.
(e)--(h) give the corresponding energy ratios $\langle\bar C_q\rangle/C_{\rm ML}$ and $\langle\bar C_q\rangle/C_{\rm tx}$
vs.\ $p$, in the same order, at 12\,dB (16-QAM) and 20\,dB (64-QAM).}
\label{fig:objective}
\end{figure}

\begin{figure}[t]
\centering
\includegraphics[width=\columnwidth]{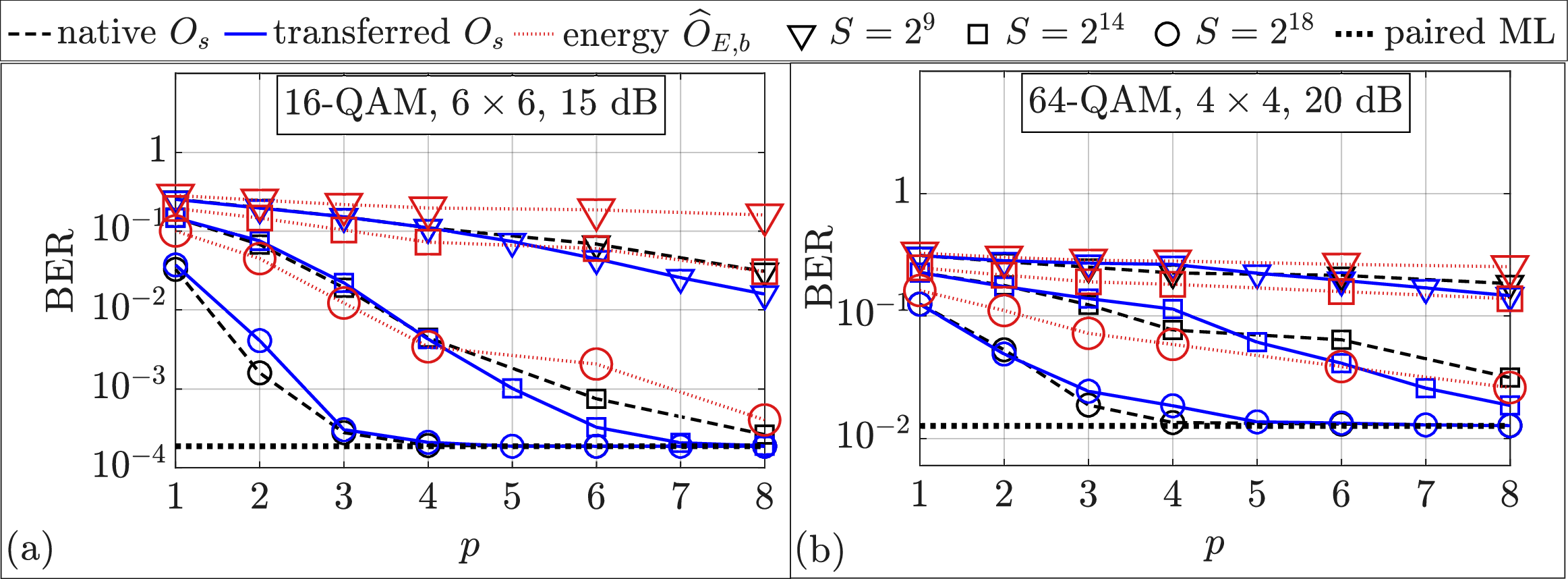}
\caption{Paired $q=24$ BER, three arms per panel: transferred $O_s$ from
$q_{\rm ref}=12$ blue solid, natively designed at $q=24$ black dashed (free
search at $p\le4$, four-parameter schedule at $p\ge6$); markers
$S=2^9,2^{14},2^{18}$ and dotted lines same-instance ML floors.
(a) 16-QAM ($6\times6$), 15\,dB; (b) 64-QAM ($4\times4$), 20\,dB; both packs
designed at 13/25\,dB, so both are off-design in SNR and the transferred pack
also in size; red dotted adds the energy-designed $\Reb$ pack
($q_{\rm ref}=12$, 15/20\,dB, at the deployed SNR).}
\label{fig:paired_q24}
\end{figure}

\begin{figure*}[!t]
\centering
\includegraphics[width=\textwidth]{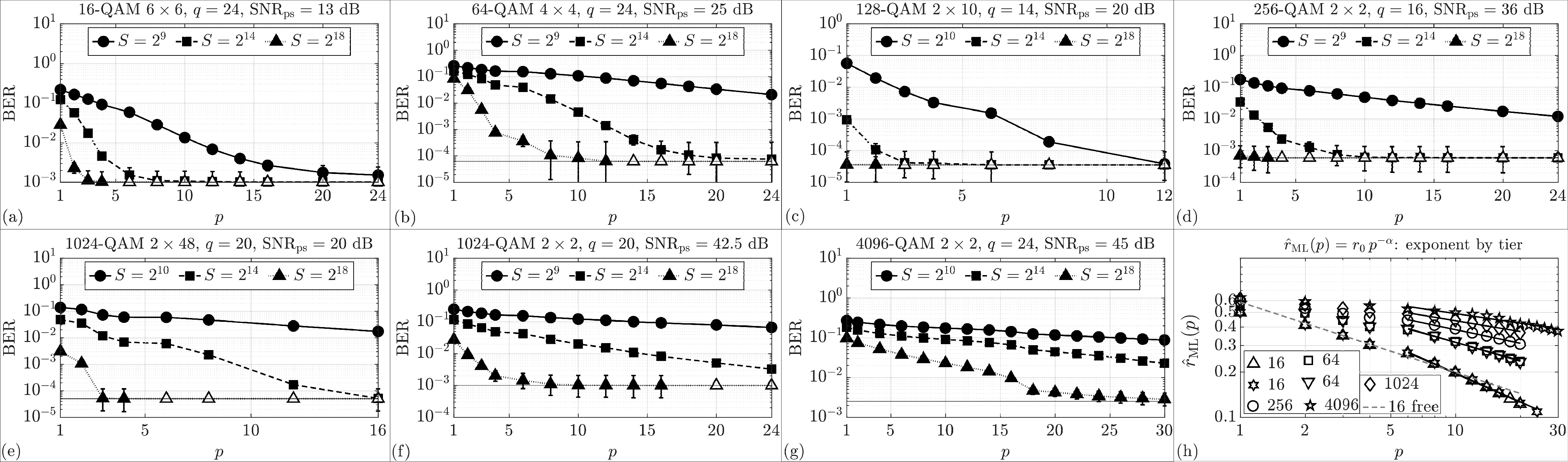}
\caption{Native scheduled \(O_s\) chains, exact-statevector best-of-$S$ BER
versus depth; thin
lines are paired ML floors, open markers coincide numerically with them, bars
are 95\% intervals. (a) 16-QAM $6\times6$, $q=24$, 13\,dB; (b) 64-QAM $4\times4$, $q=24$,
25\,dB; (c) 128-QAM $2\times10$, $q=14$, 20\,dB; (d) 256-QAM $2\times2$,
$q=16$, 36\,dB; (e) 1024-QAM $2\times48$, $q=20$, 20\,dB; (f) 1024-QAM
$2\times2$, $q=20$, 42.5\,dB; (g) 4096-QAM $2\times2$, $q=24$, 45\,dB;
(h) schedule-branch fits $\rml(p)=r_0 \, p^{-\alpha}$ by tier, with the
$6\times6$ free-arm reference dashed. Chains minimize the ML rate under the
four-parameter schedule at $p\ge6$.}
\label{fig:depth_panels}
\end{figure*}

\begin{figure}[t]
\centering
\includegraphics[width=\columnwidth]{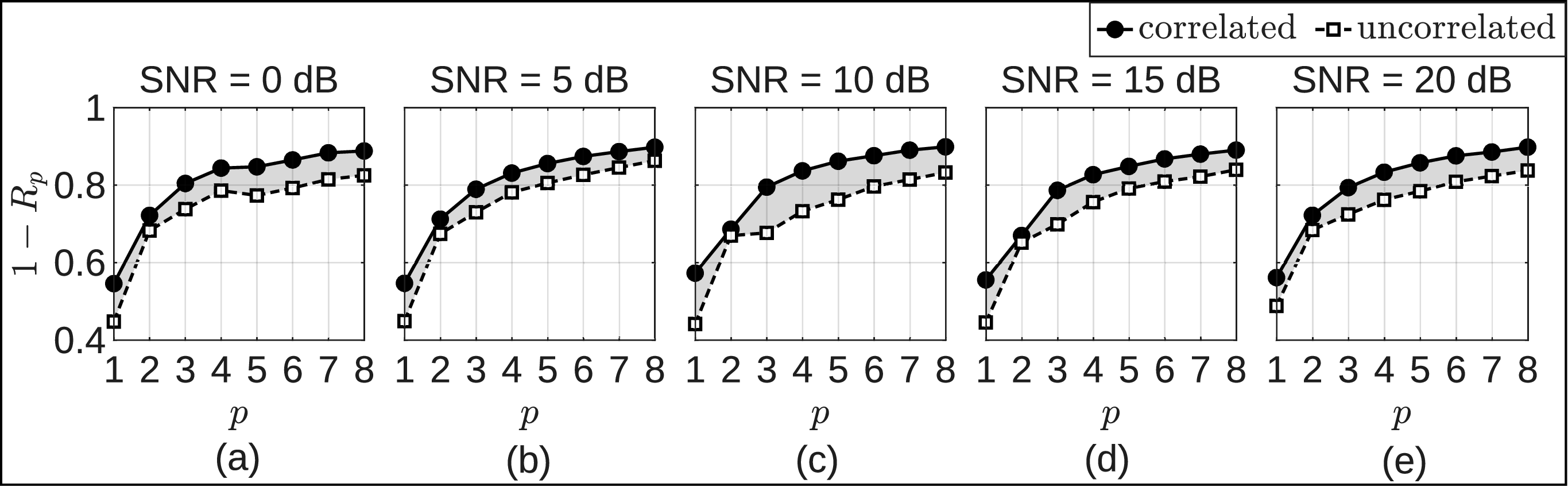}
\caption{Transferred \(\Rea\) 16-QAM packs trained at \(q_{\rm ref}=12\):
correlation-aware vs. variance-matched independent, shading their gap. (a)--(e) give $1-R_p$ of \eqref{eq:Rp}
vs. $p$ at
$q=20$   with $\mathrm{SNR}_{\rm ps}=0,5,10,15,20$\,dB.}

\label{fig:corr_depth_size}
\end{figure}

\begin{figure}[t]
\centering
\includegraphics[width=\columnwidth]{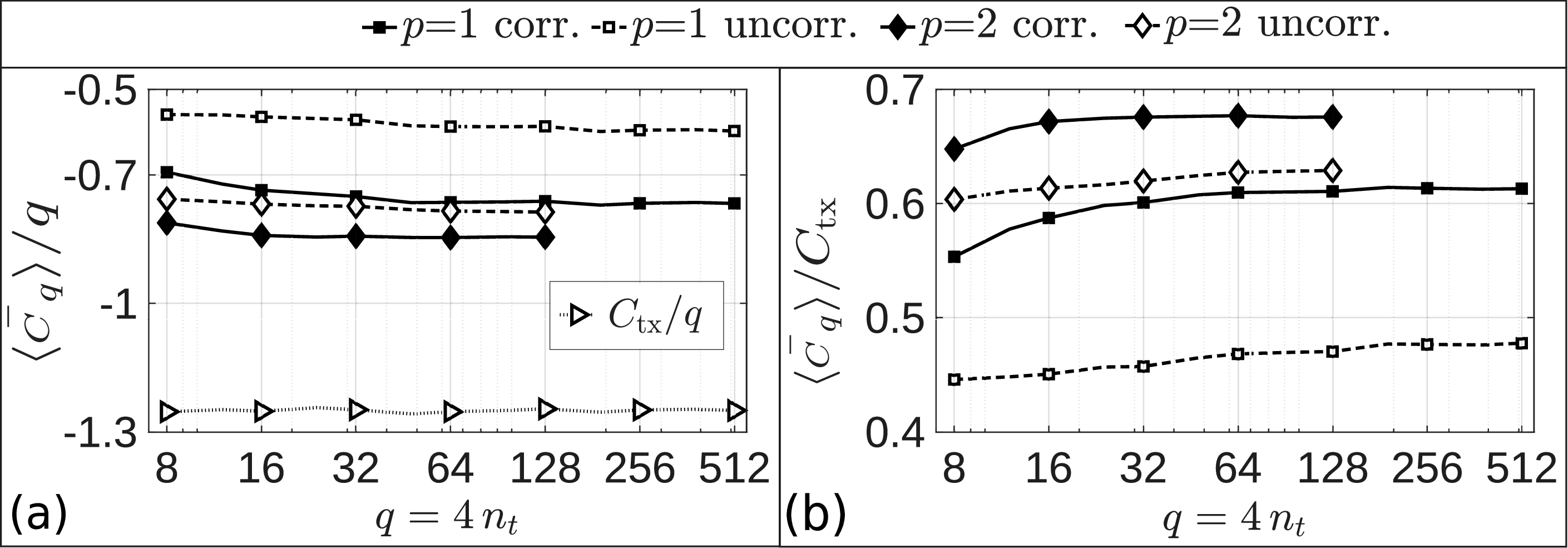}
\caption{Finite-$q$ \(\Rea\) 16-QAM expectations at
$\mathrm{SNR}_{\rm ps}=10$\,dB, correlation-aware vs. independent; exact \(p=1\) reaches $q=512$ and subset-truncated \(p=2\) reaches
$q=128$. (a) size-stable energy density $\langle\bar C_q\rangle/q$ with the
transmitted-word $C_{\rm tx}/q$ dotted; (b) ratio
$\langle\bar C_q\rangle/C_{\rm tx}$.}
\label{fig:largen}
\end{figure}

\subsubsection{Transferred versus Natively Designed \(O_s\) Chains}
\label{sec:reb_vs_os}
Fig.~\ref{fig:paired_q24} asks what transferring costs against designing
at the deployed size. Both arms minimize an ML rate on the same $2000$
instances, one shared cost table, $S=2^9,2^{14},2^{18}$, and same-instance
ML floors: the transferred arm carries the $q_{\rm ref}=12$ chain of
Section~\ref{sec:r5sims} to $q=24$, while the native arm is optimized at
$q=24$ with all $2p$ angles free for $p\le4$ and the schedule beyond.
Transfer costs almost nothing: at $S=2^{18}$ both reach the same ML floor
from $p=6$, within one standard error of it and of each other. The residual
ordering tracks the search restriction rather than the design size, the
native arm leading where it searches freely and the transferred pack where
the schedule binds.
A third arm per panel, designed on the energy objective $\Reb$ at the
same $q_{\rm ref}$ and the deployed SNR, then deployed identically ---
despite this on-SNR design advantage over the off-design sampling packs ---
separates from both: at $S=2^{18}$ it sits $16\times$ ($\EbBmid\times$)
above the transferred chain at $p=4$ in (a) ((b)), and where both sampling
arms sit on their floors by $p=8$ it remains $2.2\times$ ($\EbBend\times$)
above -- the objective, not design size or search family, sets the depth at
which ML sampling becomes reliable.

\subsubsection{\(O_s\) Engine Validation and Transfer}\label{sec:r5sims}
Table~\ref{tab:os_records} certifies the engine behind the $O_s$ curves
of Figs.~\ref{fig:objcmp} and \ref{fig:paired_q24}: the chains minimize
the frozen-reference ML rate at every depth and transfer unchanged.
Proposition~\ref{prop:planted_p1}'s exact benchmark is a consistency
check, not the shipped chain seed, and the brute-force evaluation of
\eqref{eq:planted_pp} confirms its Monte-Carlo estimate to $0.3\%$.
The matched-rate $q_{\rm ref}{=}12$ reference and deployed
Fig.~\ref{fig:objcmp}(d) fits agree within $2\%$, whereas the chain's quenched fit differs, evidence
for transfer/objective separation, not asymptotics.

\begin{table}[t]
\centering
\caption{\(O_s\) engine validation and transfer records.}
\label{tab:os_records}
\footnotesize
\setlength{\tabcolsep}{2pt}
\begin{tabular}{lll}
\hline
Check & Setting & Result\\
\hline
Planted benchmark \(\rpl(1)\), exact & \(q{=}12/20\) & \(0.47698/0.47925\)\\
\eqref{eq:planted_pp} brute vs.\ MC & \(p{=}2\), \(N{=}3{\times}10^5\) & \(0.2266/0.2260\)\\
Annealed fit, reference & \(q_{\rm ref}{=}12\) & \(0.587\,p^{-0.341}\)\\
Annealed fit, deployed & \(q{=}18\), 28\,dB & \(0.585\,p^{-0.335}\)\\
Quenched fit, reference chain & \(q_{\rm ref}{=}12\), 18\,dB &
\(0.590\,p^{-0.319}\)\\
\hline
\end{tabular}
\end{table}

\subsubsection{Exact BER--Depth Curves and the Depth Law}\label{sec:depth_law}
Exact best-of-\(S\) BER counts bit errors under an axiswise binary-reflected
Gray labeling: the \(2^{b_P}\) amplitude levels of each axis carry the bit
patterns \(g_P(k)=k\oplus\lfloor k/2\rfloor\), \(k=0,\dots,2^{b_P}-1\), so
adjacent levels differ in one bit. Native chains minimize the ML rate with
free low-depth angles then a four-parameter schedule; each simulation in
Fig.~\ref{fig:depth_panels} evaluates \eqref{eq:rs_def} at its own \(N\)
instances and size \(q\), which estimates \(\rml(p)\), and sets
\(S^\ast(p)=e^{q \, \rml(p)}\). Figs.~\ref{fig:depth_panels}(a)--(f) first reach their
paired floors at
$p=6,14,3,4,6,20$ within the measured ladders. The 64-QAM curve in (b)
crosses its independent floor between $p=16$ and 18 at $S=2^{14}$ and
reaches its paired floor from $p=14$ at $S=2^{18}$; the 4096-QAM curve in
(g) remains $1.12\times$ its floor at $p=30$ with ML-hit rate $0.997$. Predicted and
measured hit rates agree within $0.011$.

Power-law-versus-floor test on 16-QAM
$6\times6$, frozen before $p\ge4$, rejects the floor at
$\Delta\chi^2>500$. The six-point free-arm law (\(\pm\): weighted
least-squares standard error) is
\begin{equation}
  \rml(p)=0.586\,p^{-0.468\pm0.005}
  \label{eq:depthlaw66}
\end{equation}
$S^\ast(p)=S$ gives $p^\ast\approx(0.586 \, q/\ln S)^{1/0.468}=4,18,79$ for
$q=32,64,128$ at $S=2^{14}$; persistence would imply fixed-$S$ depth
$q^{2.1}$. On the common $p\ge6$ schedule branch, the seven fitted
exponents of Fig.~\ref{fig:depth_panels}(h) span $0.213$--$0.637$,
exceeding the preregistered universality threshold,
while $\alpha\log_2M=2.48$ within $6.9\%$. Two SNR-rerun pairs change
fitted slope by at most about $0.01$, so SNR mainly shifts the finite-size
floor.

\begin{table*}[t]
\centering
\caption{Symbols and Meanings}
\label{tab:symbols}
\footnotesize
\setlength{\tabcolsep}{2.0pt}
\renewcommand{\arraystretch}{0.88}
\begin{tabularx}{\textwidth}{|
>{\raggedright\arraybackslash}p{0.19\textwidth}|
>{\raggedright\arraybackslash}X||
>{\raggedright\arraybackslash}p{0.19\textwidth}|
>{\raggedright\arraybackslash}X|}
\hline
\multicolumn{2}{|l||}{\emph{System model and modulation}}&
\multicolumn{2}{l|}{\emph{Paths and scalar action}}\\
\hline
$n_t,n_r,\rho$; $M=L_RL_I,L_P$ & dimensions/aspect; constellation/PAM &
$\Acal_p,a,a^{(i)}$; $\widehat a_{\pm r},a\circ b$ & path alphabet; suffix/product\\
\hline
$b,q$; $H,\widetilde H,\widetilde y,\widetilde n$ & bits/qubits; complex/real model &
$\Phi_a,Q_\beta(a),\omega_a$ & phase, mixer weight, occupation\\
\hline
$\sigma_H^2,\sigma_n^2$; $\mathrm{SNR}_{\rm ps},\mathrm{SNR}_{\rm rx}$ & variances; SNR conventions &
$W,W^\star,\mathcal C_p$; $K_a,\bar\Phi$ & complex occupation/domain; contractions\\
\hline
$s,z,T,w_i,\varepsilon^{(\ell)}$; $i=(u_i,P_i,r_i)$ & words; map/weight; symbol-domain error; bit label &
$\Pcal_0,\Pcal_{A,\Bsf,D}$; $\Fcal^{(0)}_{A,\Bsf,0}$ & projected exponents; stationary action\\
\hline
$E_s,\Ecal_P$; $G,J_{ij},\Lcal_P$ & energies; Gram/coupling/field factor &
$\sigma_J^2,\sigma_h^2,A,\Bsf,D$ & scalar projected moments\\
\hline
\multicolumn{2}{|l||}{\emph{Ising/QAOA and metrics}}&
\multicolumn{2}{l|}{\emph{Sampling law and rates}}\\
\hline
$C,C_0,C_{\rm const}$; $J_{ij},h_i$ & Ising cost/offsets; coefficients &
$P^{\rm pl}_{q,p},\rs,\rpl,\rml$ & planted mass; design/planted/ML rates\\
\hline
$\bar J,\bar h,\bar C_q$; $U_C,U_B,B$ & normalized Hamiltonian; cost/mixer operators &
$\iota,\varpi$; $\mathcal I,q_{\rm ref}$ & Gram invariants/weight; reference set\\
\hline
$\gamma_r,\beta_r,p$; $C_{\rm ML},C_{\rm tx},C_{\rm rand}$ & angles/depth; reference costs &
$P^{\rm ML},S$; $\rml(p)=r_0p^{-\alpha}$ & ML mass/shots; fitted depth law\\
\hline
$R_p,V_p,V_p^{\rm corr}$ & residual; physical/surrogate densities &
$G_p,D_p,\Sigma$; $\mathcal L_p,\mathcal U_p$ & error Gram/kernel; lattice, support\\
\hline
$\langle\cdot\rangle$; $\E,\E_{J,h},\widehat\E_{\rm inst}$ & state/model/sample averages &
$\E_i,\E_{i<j},\E_{i\ne j},\E_{(i,j,k)}$ & uniform index-set means\\
\hline
\end{tabularx}
\end{table*}

\subsection{Energy-Objective Approximations (\(\Rea,\Reb\))}\label{sec:energy_sims}
These results are finite-size transfer evidence; they do not identify the
\(q_{\rm ref}=12\) training objective with a physical infinite-size \(V_p\) or
establish an asymptotic exponent, all-depth limit, adiabatic guarantee,
sampling law, or scalable advantage.

\subsubsection{Correlated versus Uncorrelated Angle Design}
\label{sec:corr_vs_uncorr}
The independent baseline is a pooled-variance i.i.d.-Gaussian proxy: zero-mean
Gaussian couplings and fields drawn at the pooled variances, the model the
earlier B/QPSK designs assume~\cite{gulbahar2024mimo}. It removes structured
means and species structure along with the correlations, so the gap below does
not isolate the correlation effect. Correlated and independent packs trained at
$q_{\rm ref}=12$ transfer over Table~\ref{tab:simulation_census}'s grid; at
$q=20$ the transmitted word equals ML throughout 10--20\,dB.
Fig.~\ref{fig:corr_depth_size} plots $1-R_p$ of \eqref{eq:Rp}, which is
$\langle\bar C_q\rangle/C_{\rm ML}$ by \eqref{eq:crand_zero}, and shows a
positive mean correlated-minus-independent gap at every depth, size, and SNR.
Averaged over the five SNR simulations of Fig.~\ref{fig:corr_depth_size}(a)--(e),
the $q=20$ gaps for $p=1,2,3,4$ are $0.102,0.030,0.080,0.071$, and over the
eight $q=20$ statevector points of the same simulation, the tested 16-QAM residual
decreases from $R_{p=1}\approx0.43$ to $R_{p=8}\approx0.10$.
As a check of finite-$q$ normalization and transfer (not of
Algorithm~\ref{alg:corrvp}), Fig.~\ref{fig:largen} reaches $q=512$ with
exact $p=1$ in $O(q^2)$ memory and $q=128$ with subset-truncated $p=2$:
the density in (a) stabilizes beyond $q\approx48$, and at $q=512$ the
$p=1$ ratios $\langle\bar C_q\rangle/C_{\rm tx}$ in (b) are $0.61/0.48$
(correlated/independent).

\subsubsection{Matched-Shot Placement Against Warm-Start Pipelines}
Warm-start (WS) pipelines couple a per-instance classical preprocessor
to shallow quantum layers; we reproduce a recent Gray-coded M-QAM
representative~\cite{paul2026warmstart}, namely WSLR-W, a WS linear-ramp
(LR) QAOA with Burer--Monteiro block-coordinate-descent (BM--BCD) soft
estimates, WS mixer, and printed mode decision, on identical instances
under its $\mathrm{SNR}_{\rm rx}$ convention at a common budget
$S=1024$ chosen by us. Table~\ref{tab:wslr} uses $\Rea$ without
instancewise angle optimization; ``b.o.s.'' (best of shots) pools both
ramp-rate sets. The train-once angles
track exact ML at both sizes with no per-instance classical solve, while
an ablation keeping only the warm-start product state matches the full
pipeline here: at these sizes the classical soft estimate dominates the
measured SER; this cautions against attributing gains inside warm-start
pipelines and, with decision-rule differences, prevents isolating the
quantum layers. These matched-budget numbers are not directly comparable
to \cite{paul2026warmstart}'s curves at its own settings.

\begin{table}[t]
\centering
\caption{16-QAM SER at \(p=5,S=1024\)
(\(N=3000/900\) for \(2\times2/4\times4\)).}
\label{tab:wslr}
\footnotesize
\setlength{\tabcolsep}{2.5pt}
\renewcommand{\arraystretch}{0.95}
\begin{tabular}{|l|c|c|c|c|}
\hline
 & \multicolumn{2}{c|}{$2\times2$} & \multicolumn{2}{c|}{$4\times4$} \\
Method & 12\,dB & 17\,dB & 12\,dB & 16\,dB \\
\hline
exact ML & 0.324 & 0.106 & 0.364 & 0.133 \\
\hline
\textbf{this work (offline)} & \textbf{0.324} & \textbf{0.106} & \textbf{0.366} & \textbf{0.142} \\
\hline
WSLR-W (printed mode rule)~\cite{paul2026warmstart} & 0.383 & 0.183 & 0.456 & 0.292 \\
\hline
WSLR-W (b.o.s.) & 0.330 & 0.118 & 0.388 & 0.205 \\
\hline
WS-only (no layers) & 0.331 & 0.125 & 0.384 & 0.198 \\
\hline
LR-QAOA (no warm start) & 0.578 & 0.562 & 0.843 & 0.824 \\
\hline
MMSE & 0.385 & 0.196 & 0.472 & 0.299 \\
\hline
\end{tabular}
\end{table}

\section{Conclusion}
\label{conclusion}
We formulated $M$-QAM MIMO ML detection as an Ising sampling problem for FTQ
hardware and constructed a correlation-aware, covariance-matched fixed-depth
scalar-action evaluator with an exact $O(p \,4^p)$ FWHT readout. We then
compared two offline angle-design routes: minimizing the energy readout
$V_p^{\rm corr}$ and minimizing the mean log ML-word mass. In simulations, the
sampling-designed angles exhibited the steeper finite-window ML-rate decay and
reached exact-ML BER floors at depths where the energetic design did not.
$200$-logical-qubit FTQ hardware targeted for 2029~\cite{ibm2025ftqc}
would provide the logical-qubit capacity for these depths, making offline  packs a
promising route to near-optimum detection.

\section*{Acknowledgments}
Calculations were performed partly at the T\"UB\.ITAK ULAKB\.IM High
Performance and Grid Computing Center (TRUBA) and Turkey's National Center for
High Performance Computing (UHeM), under grant 1026432026. OpenAI ChatGPT 5.6 and Anthropic Claude Fable 5 assisted with derivations, simulation, and analysis, and with manuscript editing. The author reviewed the derivations and all AI-assisted code, verified all reported results and figures, and takes full responsibility for the content of this article.

\section*{Data and Code Availability Statement}
All code/data are provided with the submission as a self-contained, pure-Python Code Ocean
compute capsule~\cite{repository}: the design and evaluation engines for both
objectives, the exact decoders, the frozen simulation records, and the scripts
that regenerate every figure.

\appendices
 
\section{}\label{app:action}

Computational-basis insertion in
$\langle\psi_p|\bar C_q|\psi_p\rangle$, $|\psi_p\rangle\triangleq|\gamma,\beta\rangle$,
leaves a common middle word
$z^{(0)}$ and slice words $z^{(\pm r)}$ (BGMZ's characteristic-function
derivative, Eq.~(D.17), and basis insertion, Eqs.~(D.21)--(D.24), in
\cite{bgmz2022constant}):
\begin{equation}
	\langle\bar C_q\rangle
	=2^{-q}\!\!\sum_{z^{(0)},\,z^{(\pm1)},\ldots,z^{(\pm p)}}
	\bar C_q(z^{(0)})\prod_{r=1}^{p}\eta_r \, \mathcal M_r^+ \, \mathcal M_r^- 
	\label{eq:path_sum}
\end{equation}
where $\mathcal M_r^+\triangleq
	\langle z^{(r+1)}|e^{-\ii\beta_rB}|z^{(r)}\rangle,\quad
	\mathcal M_r^-\triangleq
	\langle z^{(-r)}|e^{\ii\beta_rB}|z^{(-r-1)}\rangle$, $\eta_r\triangleq
	e^{-\ii\gamma_r[\bar C_q(z^{(r)})-\bar C_q(z^{(-r)})]}$ and $z^{(\pm(p+1))}=z^{(0)}$. Put $m_i=z_i^{(0)}$ and
$a^{(i)}_{\pm r}=z_i^{(\pm r)}z_i^{(\pm(r+1))}$; for any
$a\in\Acal_p=\{\pm1\}^{2p}$ let $\widehat a_{\pm r}=\prod_{\ell=r}^pa_{\pm\ell}$,
so $z_i^{(\pm r)}=m_i\widehat a^{(i)}_{\pm r}$. This is BGMZ's relative/gauge map; removing $m_i$
is exact for their sign-symmetric independent disorder, not for the Wishart
law. 
With \(s_\ell=(a_\ell+a_{-\ell})/2\), \(d_\ell=(a_\ell-a_{-\ell})/2\),
\(c_\ell=\cos\beta_\ell\) and \(\bar c_\ell=\sin\beta_\ell\), the reduced
phase and mixer weight are as follows (Eqs.~(3.8) and~(3.7) in
\cite{bgmz2022constant}, branch-conjugated):
\begin{equation}
	\Phi_a=\sum_{r=1}^p\gamma_r(\widehat a_r-\widehat a_{-r});\
	Q_\beta(a)=\prod_{\ell=1}^{p}c_\ell^{1+s_\ell}\bar c_\ell^{1-s_\ell}\ii^{\,d_\ell}
	\label{eq:path_weights}
\end{equation}

Our mixer weight is the complex conjugate of BGMZ's Eq.~(3.7) at fixed
labels (equivalently, the ket/bra branches are swapped);
$\{Q_\beta(a)\}$ are proper complex weights
(Definition~B.2 and Lemma~D.5 in \cite{bgmz2022constant}). The occupation
(the $n_a/n$ of Eq.~(D.32) in \cite{bgmz2022constant}) is
\begin{equation}
 \omega_a=\frac1q\sum_{i=1}^q\mathbf 1\big[a^{(i)}=a\big],\qquad
 \omega_a\ge0,\ \ \sum_{a\in\Acal_p}\omega_a=1
 \label{eq:occupation}
\end{equation}
Grouping paths by these counts, as in Eqs.~(D.36)--(D.37) of
\cite{bgmz2022constant}, yields after disorder averaging their multinomial
form Eq.~(4.1).
With $(a\circ b)_{\pm r}=a_{\pm r}b_{\pm r}$, the pair phase is
$\Phi_{a\circ b}$; the surrogate below drops the $m$-marks and
deterministic means, per (ii)--(iii) of Section~\ref{sec:r1obj}.

Gaussian-match the centered coefficients and set $m_i=1$, writing
$\bar J^c=(\bar J^c_{ij})_{i<j}$, $\bar h^c=(\bar h_i^c)_i$, and
$\Sigma_q\triangleq\Cov_{J,h}[\bar J^c;\bar h^c]$.
For $\xi=(\Lambda^\top,\Theta^\top)^\top$ with
$\Lambda_{ij}=\Phi_{a^{(i)}\circ a^{(j)}}=\Lambda_{ji}$ and
$\Theta_i=\Phi_{a^{(i)}}$, the characteristic function under the matched
Gaussian law is
$\E_{\rm G}[e^{-\ii(\bar J^{c\top}\Lambda+\bar h^{c\top}\Theta)}]
=e^{-\xi^{\!\top}\Sigma_q\xi/2}$.
\begin{definition}[Post-gauge projectable covariance]
\label{def:projectable}
A centered Gaussian surrogate is post-gauge projectable at fixed $p$ (motivated
by  basis grouping of Eqs.~(D.25)--(D.38) in
\cite{bgmz2022constant}) if a
functional $\Pcal_{\rm cov}$ of $\omega$ alone satisfies uniformly over \eqref{eq:occupation}:
\begin{equation}
 -\frac{1}{2q}\,\xi^{\!\top}\Sigma_q\,\xi
 \;\xrightarrow{\;q\to\infty\;}\;\Pcal_{\rm cov}(\omega) 
 \label{eq:projectable_def}
\end{equation}
\end{definition}
Configurations then reduce to
$W\in\mathcal C_p=\{W\in\mathbb C^{|\Acal_p|}:\sum_aW_a=1\}$ with exponent
$\sum_aW_a\log[Q_\beta(a)/W_a]+\Pcal_{\rm cov}(W)$, giving
\eqref{eq:free_energy}, whose readout is $V_p^{\rm corr}$
(Lemma~\ref{lem:readout}); this does not justify gauge removal, and the
uniform $-t$ shift is only a formal source. With \eqref{eq:corr_contractions},
\begin{align}
	\Pcal_{A,\Bsf,D}(W)
	&=\Upsilon_A(W)+\Upsilon_{\Bsf}(W)+\Upsilon_D(W)
	\label{eq:ABD_action}\\
	\Pcal_0(W)
	&=-\sigma_J^2\sum_{a,b}W_aW_b\Phi_{a\circ b}^2
	-\frac{\sigma_h^2}{2}\sum_a W_a\Phi_a^2
	\label{eq:P0}
\end{align}
where
$\Upsilon_A(W)\triangleq-\frac A2\sum_aW_aK_a(W)^2$,
$\Upsilon_{\Bsf}(W)\triangleq-\frac{\Bsf}{2}\bar\Phi(W)^2$, and
$\Upsilon_D(W)\triangleq-\frac D2\sum_{a,b}W_aW_b
\Phi_{a\circ b}(\Phi_a+\Phi_b)$.
With algebraic products, assume
$\Var(\bar J_{ij})=4\sigma_J^2/q$, $\Var(\bar h_i)=\sigma_h^2$,
$\Cov(\bar J_{ij},\bar J_{ik})=A/q^2$,
$\Cov(\bar h_i,\bar h_j)=\Bsf/q$, and
$\Cov(\bar J_{ij},\bar h_i)=D/q$, extending the independent-disorder
limits of Assumption~1 in \cite{bgmz2022constant} to shared-index
covariances; the diagonal terms yield \eqref{eq:P0} (Eq.~(D.61) in
\cite{bgmz2022constant} at interaction orders $r\le2$ with Gaussian
$g_r(\lambda)=-\lambda^2/2$), while uniform shared-index counting gives
$-\frac{A}{2q^3}\sum_i\sum_{j,k\ne i,\,j\ne k}
\Lambda_{ij}\Lambda_{ik}\to\Upsilon_A(W)$,
$-\frac{\Bsf}{2q^2}\sum_{i\ne j}\Theta_i\Theta_j
\to\Upsilon_{\Bsf}(W)$, and
$-\frac{D}{q^2}\sum_{i<j}\Lambda_{ij}(\Theta_i+\Theta_j)
\to\Upsilon_D(W)$, hence \eqref{eq:ABD_action}. Unconditional symmetry gives $D=0$;
conditioning in \eqref{eq:D_cond_MQAM} retains planted marks.

\section{}\label{app:vp_proofs}

\begin{proof}[Proof of Proposition~\ref{prop:planted_p1} and
\eqref{eq:planted_pp}]
Amplitude and conjugate expansion and uniform word averaging supply $2^{-2q}$;
Gaussian channel averaging via $\E\,e^{-X^\dagger DX}=\det(I+\Sigma D)^{-1}$
for $X\sim\mathcal{CN}(0,\Sigma)$ (valid for every $D$ with $\mathrm{Herm}(D)\succeq0$: whitening
$X=\Sigma^{1/2}Z$ and setting $A=\Sigma^{1/2}D\Sigma^{1/2}$ gives
$I{+}\mathrm{Herm}(A)\succ0$; factoring out $(I{+}\mathrm{Herm}\,A)^{1/2}$ and
unitarily diagonalizing the remaining skew part yields
$\det(I{+}A)^{-1}=\det(I{+}\Sigma D)^{-1}$, covering singular
$\Sigma$; convergent
since $\mathrm{Herm}(D_p)=\sigma_n^2c_pc_p^{\!\top}\succeq0$), applied
independently across the $n_r$
receivers, gives
\eqref{eq:planted_p1}. Agreement variables yield  symbol convolution,
and $p$ forward/backward replicas give $c_p,D_p,G_p$, \eqref{eq:planted_pp},
and $4p^2$ additive real coordinates.
\end{proof}
BGMZ occupation and projected exponents define the source-free action;
$\operatorname*{stat}_W$ selects Algorithm~\ref{alg:corrvp}'s branch in
$\mathcal C_p$; $(0)$ omits means, planted reference, and middle spins.
\begin{definition}[Formal scalar Gaussian source-free action]
\label{def:action}
For fixed $p,\gamma,\beta$ and positive-semidefinite covariance data with
$D=0$, the selected value is the covariance-matched extension of BGMZ's
saddle exponent (Eq.~(C.5) in \cite{bgmz2022constant}, their
$P\mapsto\Pcal_0+\Pcal_{A,\Bsf,0}$) as follows:
\begin{equation}
\Fcal_{A,\Bsf,0}^{(0)}
=\operatorname*{stat}_{W\in\mathcal C_p}
\Big[\sum_{a\in\Acal_p}W_a\log\frac{Q_\beta(a)}{W_a}
+\Pcal_0(W)+\Pcal_{A,\Bsf,0}(W)\Big]
\label{eq:free_energy}
\end{equation}
Continuation from $W=Q_\beta$ selects both branches; $0\log0\triangleq0$
on the active support; logarithms are the continuously compatible
lifts along the continuation path, not re-evaluated principal values.
\end{definition}

Stationarity gives \eqref{eq:fixed_point} (Eq.~(C.9) in \cite{bgmz2022constant};
their Eq.~(4.3) is the unit-normalized case via their Lemma~B.6), which
expands to \eqref{eq:corr_sce}  with the map defined on
	$\{W:\sum_{c}Q_\beta(c)\,e^{g_c(W)}\neq0\}$ :
\begin{equation}
W_a^\star=
\frac{Q_\beta(a)\,e^{g_a(W^\star)}}
{\sum_{c\in\Acal_p}Q_\beta(c)\,e^{g_c(W^\star)}},
\,
g_a(W)\triangleq\partial_{W_a}[\Pcal_0+\Pcal_{A,\Bsf,0}]
\label{eq:fixed_point}
\end{equation}
\begin{align}
W_a^\star&\propto Q_\beta(a)\exp\Big[
-2\sigma_J^2\textstyle\sum_b\Phi_{a\circ b}^2W_b^\star
-\tfrac{\sigma_h^2}{2}\Phi_a^2\nonumber\\
&\quad-\tfrac A2K_a(W^\star)^2
-A\textstyle\sum_bK_b(W^\star)
\Phi_{a\circ b}W_b^\star\nonumber\\
&\quad-\Bsf\,\bar\Phi(W^\star)\Phi_a\Big],
\qquad a\in\Acal_p,\quad\sum_{a\in\Acal_p}W_a^\star=1 
\label{eq:corr_sce}
\end{align}
For $A=\Bsf=0$, identifying $c_1^2=\sigma_h^2$ and
$c_2^2=2\sigma_J^2$ with the ket/bra branches conjugated ($\Pcal_0$ is then
BGMZ's Gaussian well-played $q_{\max}{=}2$ polynomial, so their
Lemma~B.6 gives unit normalization) reduces this to
BGMZ \cite[Eq.~(3.9)]{bgmz2022constant}; otherwise
correlation destroys the triangular sweep, so Algorithm~\ref{alg:corrvp}
solves $\mathsf F(W)=W-\mathsf T(W)=0$: $\mathsf T_a(W)$ are the components
of \eqref{eq:fixed_point} with its exponent $g_a$; for a direction
$w=(w_a)$, $\delta g_a$ (collected as $\delta g(W;w)$) is the derivative of
$g_a$ along $w$, and the Jacobian--vector product is the following; any failed continuation node aborts to \texttt{unresolved}:
\begin{align}
 (\mathsf F'(W)w)_a
 &=w_a-\mathsf T_a(W)\big[\delta g_a-\textstyle\sum_c
 \mathsf T_c(W)\delta g_c\big],\nonumber\\
 \delta g(W;w)
 &=-2\sigma_J^2(\Phi^2\!*\!w)-\Bsf \Phi\langle\Phi,w\rangle
 -A \big[K\!\cdot\!(\Phi\!*\!w)\nonumber\\
 &\quad+\Phi\!*\!(K\!\cdot\!w)+\Phi\!*\!(W\!\cdot\!(\Phi\!*\!w))\big]
 \label{eq:jvp}
\end{align}
Here $*$ is XOR convolution and $\cdot$ elementwise, without conjugation; a
Jacobian--vector product costs $O(p4^p)$ matrix-free.

\begin{proof}[Proof of Lemma~\ref{lem:readout}]
Let $\Fcal^{(0)}_{A,\Bsf,0}(t)$ be \eqref{eq:free_energy} with
$\Phi_a\mapsto\Phi_a-t$, $\Phi_{a\circ b}\mapsto\Phi_{a\circ b}-t$ in
$\Pcal_0+\Pcal_{A,\Bsf,0}$, so $K_a\mapsto K_a-t$, $\bar\Phi\mapsto\bar\Phi-t$
by $\sum_aW_a=1$; the source $t$ is neither QAOA time nor the homotopy
$\zeta$. At $t=0$, along the selected differentiable constrained branch,
$d_tF=\partial_tF+\sum_a(\partial_{W_a}F)\,d_tW_a^\star=\partial_tF$,
since the constrained gradient is a common Lagrange multiplier and
$\sum_ad_tW_a^\star=0$ (conditional on a differentiable branch
selection; cf.\ the saddle-value (real) envelope theorem, Thm.~4 in
\cite{milgrom2002envelope}); the explicit
$t$-derivative gives $V_p^{\rm corr}$ of \eqref{eq:Vpcorr}.
\end{proof}

\end{document}